\documentclass[reqno,11pt]{amsart}

\usepackage[utf8]{inputenc}

\usepackage{amsmath}
\usepackage{amsfonts}
\usepackage{amssymb}
\usepackage{amsthm}
\usepackage{amscd}

\usepackage[mathscr]{eucal}
\usepackage{mathtools}

\usepackage{physics}
\usepackage{tensor}
\usepackage{slashed}

\usepackage{bbold}
\usepackage{scalerel}

\usepackage[nolist]{acronym}
\usepackage[linktocpage]{hyperref}
\usepackage[font=small,labelfont=bf]{caption}
\usepackage{enumerate}
\usepackage{enumitem}

\usepackage{graphicx}
\usepackage{pgfplots}
\usepackage{stmaryrd}

\pgfplotsset{compat=newest}

\DeclareMathOperator{\supp}{supp}

\DeclareMathOperator{\Span}{span}
\DeclareMathOperator{\id}{id}

\DeclareMathOperator{\sgn}{sgn}

\DeclareMathOperator{\End}{End}
\DeclareMathOperator{\Lin}{L}

\DeclareMathOperator{\diag}{diag}
\DeclareMathOperator{\vol}{vol}

\newcommand{\defeq}{\coloneqq}
\newcommand{\eqdef}{\eqqcolon}
\newcommand{\isospectralto}{\asymp}

\newcommand{\complexconj}[1]{\overline{#1}}
\newcommand{\disjointUnion}{\dot{\bigcup}}

\newcommand{\spinbra}[1]{\stretchrel*{\prec}{\left. #1 \right|} \left.\hspace{-0.1cm} #1 \right|}
\newcommand{\spinket}[1]{\left| #1 \right.\hspace{-0.1cm} \stretchrel*{\succ}{\left| #1 \right.}}
\newcommand{\spinbraket}[2]{{\stretchrel*{\prec}{\left. #1 \middle| #2 \right.} \left.\hspace{-0.1cm} #1 \middle| #2  \right.\hspace{-0.1cm} \stretchrel*{\succ}{\left. #1 \middle| #2 \right.}}}

\newcommand{\N}{\mathbb{N}}
\newcommand{\Z}{\mathbb{Z}}
\newcommand{\R}{\mathbb{R}}
\newcommand{\C}{\mathbb{C}}
\renewcommand{\H}{\mathcal{H}}
\newcommand{\F}{\mathcal{F}}
\renewcommand{\L}{\mathcal{L}}

\numberwithin{equation}{section}

\newtheorem{definition}{Definition}[section]
\newtheorem{theorem}[definition]{Theorem}
\newtheorem{proposition}[definition]{Proposition}
\newtheorem{lemma}[definition]{Lemma}
\newtheorem{corollary}[definition]{Corollary}
\newtheorem{example}[definition]{Example}

\theoremstyle{plain}

\begin{acronym}
    \acro{qft}[QFT]{quantum field theory}
    \acro{aqft}[AQFT]{algebraic approach to quantum field theory}
    \acro{cars}[CARs]{canonical anti-commutation relations}
    \acro{el}[EL]{Euler-Lagrange}
\end{acronym}

\begin{document}
    \title[The Continuum Limit in Curved Spacetimes]{The Continuum Limit Analysis of Causal Fermion Systems for Curved Spacetimes}

    \author[F. Finster]{Felix Finster}
    \author[P. Fischer]{Patrick Fischer \\ \\ October 2026}
    \address{Fakultät für Mathematik, Universität Regensburg, D-93040 Regensburg, Germany}
    \email{finster@ur.de, patrick.fischer@ur.de}

    \begin{abstract}
        We construct the causal fermion system for globally hyperbolic spacetimes starting in the framework of algebraic
        quantum field theory.
        The fermionic projector is identified with the one-particle density operator of a quasi-free Hadamard state.
        The ultraviolet regularization is built into the fermionic projector via a chart-independent $i\varepsilon$-regularization scheme.

        The continuum limit analysis is developed in globally hyperbolic spacetimes.
        It is shown that the Euler-Lagrange equations of the causal action principle are satisfied in this setup if and
        only if the coupled Einstein-Dirac equations hold.
    \end{abstract}

    \maketitle
    \tableofcontents

    \section{Introduction}\label{sec:introduction}

    The formulation of \ac{qft} on a general globally hyperbolic spacetime presents a fundamental conceptual
    challenge: the absence of a global timelike Killing vector field deprives the theory of a preferred vacuum state and
    a canonical notion of particles~\cite[Section 6.2]{Fewster_2012}~\cite{wald2009formulationquantumfieldtheory}.
    Therefore, the physical description in terms of the notions used in \ac{qft} becomes observer-dependent.

    To overcome these problems, the \ac{aqft} shifts the fundamental focus from Hilbert spaces to the algebra of observables.
    In this framework, the observer-dependent physical configurations are described by states on the algebra.
    For fermionic fields, the algebra encodes the \ac{cars} and non-interacting systems are entirely characterized by
    quasi-free states (for details, see~\cite{wald2009formulationquantumfieldtheory, bratteli-robinson-1997}).

    While \ac{aqft} effectively handles the observer-dependence of states, it still relies on a fixed spacetime.
    Another approach to fundamental physics is the theory of causal fermion systems.
    Rather than treating spacetime as a fixed background upon which quantum fields are defined, causal fermion
    systems propose a framework where the fermionic wave functions themselves are the fundamental entities.
    The corresponding action, the \textit{causal action}, is then defined for a measure on a set of operators acting on
    these wave functions.
    This shift of perspective provides a candidate for a unifying quantum theory and general relativity.
    The spacetime, including the spin structure, is then completely encoded in the measure.
    Therefore, it is a dynamical entity of the theory.

    A priori, the theory of causal fermion systems does not impose any specific topology of the underlying spacetime.
    In particular, it does not require a smooth manifold structure.
    However, to connect this abstract mathematical framework to conventional physics, one considers a regime where the
    spacetime of a causal fermion system can be identified with a smooth manifold.
    This is achieved through the \textit{continuum limit}, wherein a microscopic regularization parameter $\varepsilon$
    is taken to zero.
    This continuum limit has been studied in great detail for the example of Minkowski spacetime (e.g.~\cite{cfs}).
    In that flat setting, this provides a mechanism to derive the standard model as an effective field theory directly
    from the causal action.

    The construction of the causal fermion system for Minkowski spacetime relies heavily on the symmetries of the spacetime,
    in particular, on the existence of a global timelike Killing vector field.
    Thus, the analysis of general relativity is restricted to linearized gravity as a perturbation around the Minkowski metric.
    In this paper, we present a construction of causal fermion systems for general globally hyperbolic spacetimes
    and extend the derivation of the Einstein equations to the full non-linear setting.

    To achieve this, we start with the general setup of \ac{aqft}, where a quasi-free Hadamard state is used to define
    the vacuum.
    This state then defines a Hilbert space of fermionic wave functions through its corresponding one-particle density
    operator.
    This construction provides the Hilbert space for the causal fermion system (see Definition~\ref{def:causal-fermion-system}).
    In addition, the constructed causal fermion system does not rely on a particular choice of the quasi-free Hadamard state.
    Thus, even though the quasi-free Hadamard state is not uniquely defined for general globally hyperbolic spacetimes,
    the resulting linearized field equations are independent of this choice.

    In addition, the corresponding entity of the causal fermion system (more precisely, the fermion projector) inherits
    the specific short-distance singularity structure of the state.
    This structure is commonly studied using the Hadamard parametrix.
    Equivalently, the singularity structure is also encoded in the Schwinger-DeWitt expansion, which allows for an easier
    analysis in the context of causal fermion systems.
    We also show that the Schwinger-DeWitt expansion corresponds to the light cone expansion, which is the standard tool
    for the analysis of causal fermion systems in the example of Minkowski spacetime~\cite{light, firstorder}.
    Therefore, the construction of the causal fermion system for a general globally hyperbolic spacetime is a direct
    generalization of the case of Minkowski spacetime.

    After the construction of the causal fermion system for globally hyperbolic spacetimes, we show that, when considering
    only the leading term in the Schwinger-DeWitt expansion (most singular contribution), every globally hyperbolic
    spacetime is a critical point in causal action.
    Moreover, the next-to-leading contributions then give rise to the main theorems, Theorem~\ref{thm:trace-free-vacuum-einstein-equations}
    and Theorem~\ref{thm:trace-free-einstein-equations-with-matter}, of this paper,
    both of which are derived exclusively from the causal action principle.
    By the Bianchi identity, these theorems then directly imply
    \begin{enumerate}[label=(\arabic*), font=\upshape]
        \item Corollary~\ref{cor:vacuum-einstein-equation}: The vacuum Einstein equations,
        \item Corollary~\ref{cor:einstein-matter-equation}: The Einstein equations coupled to fermionic matter.
    \end{enumerate}

    The paper is organized as follows.
    In Section~\ref{sec:preliminaries}, we introduce the necessary geometric preliminaries and outline the derivation of
    the Schwinger-DeWitt expansion for Klein-Gordon-type operators, alongside a brief review of the causal action principle.
    Section~\ref{sec:iepsilon-regularization-of-the-schwinger-dewitt-expansion} addresses the required regularization of
    the singularities on the light cone by introducing a chart-independent $i\varepsilon$-regularization scheme based on a regularizing
    scalar field.
    In Section~\ref{sec:causal-fermion-systems-for-globally-hyperbolic-spacetimes}, we construct the causal fermion system
    for globally hyperbolic spacetimes.
    Finally, in Section~\ref{sec:perturbative-analysis-of-the-euler-lagrange-equations}, we perform a perturbative analysis
    of the restricted \ac{el} equations.
    By studying the next-to-leading order geometric contributions, we show that the causal action includes the trace-free
    part of the classical vacuum Einstein equations, and subsequently, the trace-free part of the Einstein equations coupled to
    matter through the inclusion of fermionic perturbations.

    \section{Preliminaries}\label{sec:preliminaries}

    Throughout this paper, we refer to \textit{classical spacetime} as a globally hyperbolic Lorentzian spin manifold $(M, g)$
    of dimension $d > 2$ with metric signature $(+, -, -, \dots)$.
    For Riemann tensor we use the convention
    \begin{align}
        \tensor{R}{^\rho _\sigma _\mu _\nu} = \partial_\mu \Gamma^\rho_{\nu \sigma} - \partial_\nu \Gamma^\rho_{\mu \sigma} + \Gamma^\rho_{\mu \lambda} \Gamma^{\lambda}_{\nu \sigma} - \Gamma^\rho_{\nu \lambda} \Gamma^{\lambda}_{\mu \sigma} \,.
    \end{align}
    Correspondingly, the Ricci tensor and the scalar curvature are given by $R_{\mu \nu} = \tensor{R}{^\rho _\mu _\rho _\nu}$ and $R = g^{\mu \nu} R_{\mu \nu}$ respectively.
    We denote the spinor bundle over $M$ by $SM$ with fibers $S_x M$ at $x \in M$ (for details, see~\cite{baum, lawson+michelsohn}).

    To keep this paper self-contained, we briefly outline the geometric preliminaries required for the Dirac equation in curved spacetimes.
    This includes basic properties of Synge's world function, the van Vleck-Morette determinant, and the spin parallel transport.
    Next, we outline the derivation of the Schwinger-DeWitt expansion for the Green's function.
    Finally, we introduce the framework of causal fermion systems.

    \subsection{The Dirac Equation in Curved Spacetimes (Part I)}\label{subsec:the-dirac-equation-in-curved-spacetimes-1}

    On each fiber at $x \in M$ of the spinor bundle $SM \to M$, the Clifford multiplication is represented by Dirac matrices
    $\gamma_x: T_x M \to \End(S_x M)$ satisfying the anti-commutation relation
    \begin{align}
        \left\{ \gamma_x(u), \gamma_x(v) \right\} = 2 g_x(u, v)\,\id_{S_x M} \quad \text{ for all } u, v \in T_x M \,.
    \end{align}
    We refer to a connection that is compatible with the metric and Clifford multiplication as a \textit{spin connection}.
    For each spinor bundle, there exists a unique, purely geometric spin connection~\cite{baum, friedrich}~\cite[Chapter 4]{review}.

    \begin{lemma}
        On $SM \to M$, there exists a unique spin connection $\nabla$ that is compatible with the metric and Clifford multiplication.
    \end{lemma}

    Note that the uniqueness is only given if one restricts to spin connections without any gauge potential.
    For example, for a vector potential $A$, $D_{\mu} \defeq \nabla_\mu - i A_\mu$ defines another valid spin connection.
    For now, we only consider the unique geometric spin connection $\nabla$.
    However, we want to emphasise that the following considerations easily generalize to arbitrary spin connections.
    The corresponding Laplacian for a spin connection $\nabla$ is defined as
    \begin{align}
        \Box^S = \nabla_\mu \nabla^\mu = -\nabla^* \nabla \,.
    \end{align}
    On~$(M, g)$, the \textit{Dirac equation} is defined as
    \begin{align}
        \label{eq:dirac-equation}
        \left(i \gamma^\mu \nabla_\mu - m\right) u = 0
    \end{align}
    for a spinor field $u$ and mass parameter $m > 0$.
    A spinor field is a smooth section of the spinor bundle $SM \to M$.
    We denote the set of smooth sections by $\Gamma(M, SM)$.

    In addition to the Dirac equation, there exists a corresponding spinorial Klein-Gordon equation.
    The relation follows from the Lichnerowicz formula~\cite{lichnerowicz, schroedinger}.

    \begin{proposition}[Lichnerowicz Formula]
        Let $\nabla$ be the geometric spin connection on the spinor bundle $SM$, then
        \begin{align}
            \left(i \gamma^\mu \nabla_{\mu}\right)^2 = -\Box^S + \frac{R}{4} \,,
        \end{align}
        where $R$ is the scalar curvature of $g$ and $\Box^S = \nabla_\mu \nabla^\mu$ is the connection Laplacian.
    \end{proposition}

    Hence, any solution of the Dirac equation~\eqref{eq:dirac-equation} also satisfies the spinorial Klein-Gordon equation
    \begin{align}
        \label{eq:spin-klein-gordon}
        \left( -\Box^S + \frac{R}{4} - m^2 \right) u = 0 \,.
    \end{align}
    In addition, if $u$ is a solution of~\eqref{eq:spin-klein-gordon} then
    \begin{align}
        \label{eq:klein-gordon-to-dirac}
        \tilde{u} \defeq \left( i \gamma^\mu \nabla_\mu + m \right) u
    \end{align}
    is a solution of the Dirac equation~\eqref{eq:dirac-equation}.
    In particular, every solution of the Dirac equation can be written in the form of~\eqref{eq:klein-gordon-to-dirac},
    because $i \gamma^\mu \nabla_\mu + m$ corresponds to a multiplication by $2m > 0$ on the space of solutions.
    With this in mind, we restrict our attention to the spinorial Klein-Gordon equation~\eqref{eq:spin-klein-gordon} in Section~\ref{subsec:the-schwinger-dewitt-expansion}.

    In preparation, we introduce some necessary geometric objects on $(M, g)$ and outline their properties, which are needed
    for the Schwinger-DeWitt expansion.
    For the remaining part of this section, let~$U$ be a \textit{normal neighborhood},
    meaning that for all $x, y \in U$ there exists a unique geodesic $\gamma : [0, 1] \to M$ connecting $x$ and $y$.
    In addition, we require $U$ to be \textit{convex}, i.e.~for every $x, y \in U$, if $\gamma: [0, 1] \to M$ is the unique geodesic
    connecting $x$ and $y$, then $\gamma(t) \in U$ for all $t \in [0, 1]$.
    \textit{Synge's world function} is then locally defined as half the signed squared geodesic distance between $x$ and $y$,
    with the sign chosen such that it is positive in timelike and negative in spacelike directions.
    Equivalently, we have the following definition for Synge's world function~\cite[Section 3.1]{Poisson_2011}.

    \begin{definition}[Synge's World Function]
        \textbf{Synge's world function} $\sigma$ is locally defined by
        \begin{align}
            \sigma(x, y) \defeq \frac{1}{2} \int_0^1 g_{\gamma(s)}(\dot{\gamma}(s), \dot{\gamma}(s)) ~\dd s \quad \text{ for } \quad x, y \in U\,
        \end{align}
        where $U$ is a convex normal neighborhood and $\gamma: [0, 1] \to M$ is the unique (affinely parametrized) geodesic connecting
        $x = \gamma(0)$ and $y = \gamma(1)$.
    \end{definition}

    At this point, we want to highlight that the derivative of $\sigma(x, y)$ with respect to the first argument is an
    element in $T_x M$, whereas the derivative with respect to the second argument is an element in $T_y M$.
    To simplify the notation and following the standard convention in the literature~(for example~\cite{Poisson_2011}),
    we denote the components of $\nabla^{(1)} \sigma(x, y)$ by $\sigma^\mu$ and the components of $\nabla^{(2)} \sigma(x, y)$ by $\sigma^{\nu'}$.
    The prime of the index indicates that the derivative is taken with respect to the second argument, and therefore
    it is an element of $T_y M$.
    For the chart-independent formulations, we adhere to the notations~$\nabla^{(1)} \sigma$ and $\nabla^{(2)} \sigma$, respectively.

    \begin{proposition}[Fundamental Bi-Tensor Identity]\label{prop:fundamental-bi-tensor-identity}
        Let $U$ be a convex normal neighborhood and $x, y \in U$.
        Then Synge's world function $\sigma(x, y)$ satisfies the relations
        \begin{align}
            \label{eq:fundamental-bi-tensor-identity}
            2 \sigma(x, y) &= g_x \left( \nabla^{(1)} \sigma(x, y), \nabla^{(1)} \sigma(x, y) \right) \nonumber \\
            &= g_y \left( \nabla^{(2)} \sigma(x, y), \nabla^{(2)} \sigma(x, y) \right) \,.
        \end{align}
    \end{proposition}

    For the proof, we refer to~\cite[Section 3.3]{Poisson_2011}.
    In a chart, the fundamental bi-tensor identity reads
    \begin{align}
        2 \sigma = g_{\mu \nu} \sigma^{\mu} \sigma^{\nu} = g_{\mu' \nu'} \sigma^{\mu'} \sigma^{\nu'} \,.
    \end{align}
    In the above expression, we suppressed the arguments $(x, y)$ of $\sigma$ and its derivatives.
    From the notation using primed and unprimed indices, it is clear that $g_{\mu \nu}$ has to be evaluated at the first
    argument of $\sigma$, whereas $g_{\mu' \nu'}$ is evaluated at the second argument of $\sigma$.

    By differentiating the fundamental bi-tensor identity~\eqref{eq:fundamental-bi-tensor-identity} one obtains the so-called
    \textit{eikonal identity}
    \begin{align}
        \Lambda_{y, x} \nabla^{(1)} \sigma(x, y) = -\nabla^{(2)} \sigma(x, y)
    \end{align}
    where $\Lambda_{y, x}: T_x M \to T_y M$ is a linear transport map along the geodesic from $x$ to $y$.
    In a chart, the matrix components of the transport map $\Lambda_{y, x}$ are given by
    \begin{align}
        \label{eq:geodesic-transport-map}
        \left( \Lambda_{y, x} \right)^{\nu'}_{\mu} = - \nabla^{(2) \nu'} \sigma_\mu \,.
    \end{align}
    We now define the van Vleck-Morette determinant in terms of the transport map as follows~\cite{moretti}.

    \begin{definition}[Van Vleck-Morette Determinant]
        The \textbf{van Vleck-Morette determinant} $\Delta$ is locally defined by
        \begin{align}
            \Delta(x, y) \defeq (-1)^d \frac{\det \left[\Lambda_{y, x}\right]}{\sqrt{g(x)} \sqrt{g(y)}} \,,
        \end{align}
        where $g(x) = \abs{\det g_{\mu \nu}(x)}$ and $\Lambda_{y, x}$ is given by~\eqref{eq:geodesic-transport-map}.
    \end{definition}

    \begin{proposition}
        Let $\Delta$ be the van Vleck-Morette determinant.
        Then it satisfies the equation
        \begin{align}
            \sigma^{\mu} \nabla^{(1)}_{\mu} \Delta^\frac{1}{2} = \frac{1}{2} \left( d - \Box^{(1)} \sigma \right) \Delta^\frac{1}{2}\,,
        \end{align}
        where $d = \dim M$ is the dimension of the spacetime.
    \end{proposition}

    Another necessary geometric object is the \textit{spin parallel transport} $U(x, y): S_x M \to S_y M$ along the geodesic
    connecting $x$ and $y$.
    It is defined as the unique solution of the transport equation~\cite[Chapter 17]{dewitt2}
    \begin{align}
        \label{eq:spin-parallel-transport}
        \nabla^{(1)}_\mu U(x, y) = 0 \quad \text{ with } \quad U(x, x) = \id_{S_x M} \,.
    \end{align}
    The spin parallel transport map has the property
    \begin{align}
        U(x, y) U(y, z) = U(x, z) \,,
    \end{align}
    which also implies that $U(x, y) U(y, x) = \id_{S_x M}$.

    This completes the necessary geometric preliminaries.
    In the next section, we derive an ansatz for the Green's functions and Wightman bi-solutions of the spinorial
    Klein-Gordon equation~\eqref{eq:spin-klein-gordon}.

    \subsection{The Schwinger-DeWitt Expansion}\label{subsec:the-schwinger-dewitt-expansion} \hspace{0.5cm} % hspace to avoid overflow of Green's, might be removed
    The singular structure of Green's functions and bi-solutions of Klein-Gordon-type operators is a central object of
    study in quantum field theory in curved spacetime.
    In the limit $\sigma(x,y)\to 0$, these distributions develop characteristic short-distance divergences.
    One standard way to encode this behaviour is via the \textit{Hadamard parametrix} (e.g.~see~\cite{hadamardoriginal}),
    which expresses the Green's function in terms of the world function $\sigma$, the van Vleck-Morette determinant,
    and a sequence of smooth coefficients obtained by solving transport equations.

    An equivalent description is provided by the \textit{Schwinger-DeWitt expansion}.
    This expansion goes back to Schwinger’s proper-time representation~\cite{schwinger51} and was then developed systematically
    by DeWitt~\cite{dewitt2}.
    This expansion has the technical advantage that the solutions of the transport equations satisfy simple recursive relations,
    which are particularly useful in the context of fermionic fields.

    In this section, we outline the derivation of the Schwinger-DeWitt expansion for the bi-solution of a
    Klein-Gordon-type operator acting on sections of a general vector bundle $E \to M$.
    This general framework will later be applied to the spinorial Klein-Gordon operator \eqref{eq:spin-klein-gordon}.
    For Riemannian manifolds and Laplace-type operators with arbitrary endomorphism-valued potentials, the systematic
    treatment is due to Avramidi~\cite{avramidi2002heat}.
    The extension to Lorentzian signature and to general Klein-Gordon-type operators on vector bundles follows from the
    same arguments.
    To keep the paper self-contained, we now outline the derivation and the required transport identities
    in the general bundle-valued Lorentzian setting.

    Let $E \to M$ be a vector bundle over $(M, g)$ with a metric compatible connection $\nabla$.
    Then, a Klein-Gordon equation is a second-order differential equation of the form
    \begin{align}
        (-\Box^E + V(x) - m^2) u(x) = 0 \quad \text{ for } u \in \Gamma(M, E)
    \end{align}
    with potential $V(x) \in \End(E_x)$ and connection Laplacian $\Box^E = \nabla_\mu \nabla^\mu$.

    We are interested in fundamental solutions~$G(x,y)$ of \textit{Hadamard form} for the Klein-Gordon operator.
    More precisely, they are distributional solutions of the equation
    \begin{align}
        \left( -\Box^{(E, 1)} + V(x) - m^2 \right) G(x, y) = 0
    \end{align}
    with a singularity structure on the light cone of the form
    \begin{align}
        \label{eq:hadamard-form}
        G(x,y) = \lim_{\varepsilon \searrow 0} \left( \frac{U(x,y)}{\sigma^\varepsilon(x,y)^p}
                                                   + V(x,y) \, \log \sigma^\varepsilon(x,y)
                                                   + W(x,y) \right) \,.
    \end{align}
    Here, $U$, $V$, and $W$ are smooth functions, and $p \in \N/2$ depends on the spacetime dimension.
    In odd dimensions, the power $p$ takes half-integer values, causing $\sigma(x,y)^p$ to develop a branch cut for spacelike
    separations where $\sigma(x,y) < 0$.
    To address this and regularize the poles, we introduce $\sigma^\varepsilon$, an analytic continuation of the Synge world
    function into the complex plane.
    This $i\varepsilon$-regularization unambiguously dictates the branch choice and ensures the distributions remain well-defined
    across the light cone.
    More specifically,
    \begin{align}
        \sigma^\varepsilon(x,y) = \sigma(x,y) - i \varepsilon f(x,y) \:,
    \end{align}
    where~$f(x,y)$ can be chosen for example as~$f(x,y) = T(y) - T(x)$ with~$T$
    an arbitrary time function (this regularization will be discussed in more detail in
    Section~\ref{sec:iepsilon-regularization-of-the-schwinger-dewitt-expansion}).

    The Schwinger-DeWitt expansion gives a systematic computational procedure
    to express $G(x, y)$ as a formal series in terms of Bessel functions\footnote{
        Originally, Schwinger and DeWitt derived the expansion for the Green's function~\cite{dewitt2, schwinger51}.
        Since the expansion only captures the singularity structure, it also holds for causal fundamental solutions.
    }.

    To analyze them systematically, we introduce the so-called light cone expansion symbols $T^{(n)}(x, y)$\footnote{
        The notation follows the general convention of causal fermion systems in the continuum limit~\cite{cfs}, where
        the $T^{(n)}$ are used in the light cone expansion of the fermionic projector.
    }.

    \begin{definition}[Light Cone Expansion Symbols]
        \label{def:light-cone-expansion-symbols}
        Let $m > 0$, $U \subseteq M$ be a convex normal neighborhood and $x, y \in U$.
        Then, for any $n \in \Z$ the \textbf{light cone expansion symbols} are defined as
        \begin{align}
            T^{(n)}(x, y) \defeq - \frac{(-2m^2)^{\nu_n}}{(2\pi)^{\frac{d}{2} + 1}} \lim_{\varepsilon \searrow 0} \frac{K_{\nu_n} \left( z^\varepsilon \right)}{(z^\varepsilon)^{\nu_n}} \quad \text{ with } \quad \nu_n = \frac{d}{2} - 1 - n \,,
        \end{align}
        where $d = \dim M$, $z^\varepsilon = m \sqrt{-2 \sigma^\varepsilon(x, y)}$ and $K_\nu$ are the modified Bessel functions of second kind.
    \end{definition}

    In standard \ac{qft}, the light cone expansion symbol $T^{(n)}$ correspond the Wightman two-point function of a scalar field
    on a $\tfrac{d}{2} - n$ dimensional spacetime.
    However, the Wightman functions are usually normalized by the factor
    \begin{align*}
        \frac{1}{(2 \pi)^{\tfrac{d}{2} - n}} \,.
    \end{align*}
    In the context of causal fermion systems, the light cone expansion symbols are used the expand the kernel of the fermionic
    projector, which is normalized as a projection (for more details, see~\cite[Section 15.4]{intro}).

    For notational convenience, the parameter~$\varepsilon$ and the limit~$\varepsilon \searrow 0$ will be omitted in the following calculations.
    The specific functional form of the distributions $T^{(n)}$ is motivated by flat spacetime, where the modified Bessel
    functions of the second kind, $K_\nu$, constitute the fundamental causal solutions to the massive wave equation.
    By adopting these flat-space solutions as a basis, the Schwinger-DeWitt expansion systematically isolates the purely
    geometric deviations into the smooth coefficients.
    This is mainly due to the following relations of the light cone expansion symbols.

    \begin{lemma}
        \label{lem:properties-of-light-cone-expansion-symbols}
        The light cone expansion symbols $T^{(n)}$ satisfy the following identities:
        \begin{enumerate}[label=(\arabic*), font=\upshape]
            \item $\displaystyle \quad \nabla^{(1)}_\mu T^{(n)} = - \frac{\sigma_\mu}{2} T^{(n - 1)}$, \\
            \item $\displaystyle \quad -\frac{\sigma}{2} T^{(n - 1)} = \left( n + 1 - \frac{d}{2} \right) T^{(n)} + m^2 T^{(n + 1)}$.
        \end{enumerate}
    \end{lemma}
    \begin{proof}
        \textrm{(1):} Note that $z = m \sqrt{-2 \sigma}$ satisfies $\nabla^{(1)}_\mu z = - \frac{m^2 \sigma_\mu}{z}$.
        Thus, for any function $f(z)$, we have
        \begin{align*}
            \nabla^{(1)}_\mu f(z) = m^2 \sigma_\mu \left( - \frac{1}{z} \dv{z} \right) f(z)\,.
        \end{align*}
        Applying relation~\cite[eq.~(10.29.4)]{DLMF} of the modified Bessel functions, we obtain
        \begin{align*}
            \nabla^{(1)}_\mu T^{(n)} &= \frac{(-2m^2)^{\nu_n}}{16 \pi^3} \nabla^{(1)}_{\nu} \left( \frac{K_{\nu_n}(z)}{z^{\nu_n}} \right) \\
            &= \frac{(-2m^2)^{\nu_n}}{16 \pi^3} ~ m^2 \sigma_\mu \left( - \frac{1}{z} \dv{z} \right) \left(\frac{K_{\nu_n}(z)}{z^{\nu_n}} \right) \\
            &= -\frac{\sigma_\mu}{2} \frac{(-2m^2)^{\nu_n + 1}}{16 \pi^3} \frac{K_{\nu_n + 1}(z)}{z^{\nu_n + 1}} = -\frac{\sigma_\mu}{2} T^{(n - 1)} \,.
        \end{align*}

        \noindent
        \textrm{(2):} From the relation~\cite[eq.~(10.29.1)]{DLMF}
        \begin{align*}
            K_{\nu_n + 1}(z) = \frac{2 \nu_n}{z} K_{\nu_n}(z) + K_{\nu_n - 1}(z)
        \end{align*}
        of the modified Bessel functions, we get
        \begin{align*}
            -\frac{\sigma}{2} T^{(n - 1)} &= -\frac{z^2}{4m^2} \frac{(-2m^2)^{\nu_n + 1}}{16 \pi^3} \frac{K_{\nu_n + 1}(z)}{z^{\nu_n + 1}} \\
            &= \frac{(-2m^2)^{\nu_n}}{16 \pi^3} \left( -\nu_n \frac{K_{\nu_n}(z)}{z^{\nu_n}} - \frac{1}{2} \frac{K_{\nu_n - 1}(z)}{z^{\nu_n - 1}} \right) \\
            &= -\nu_n T^{(n)} + m^2 T^{(n + 1)} \,. \qedhere
        \end{align*}
    \end{proof}

    \begin{corollary}
        \label{cor:properties-of-light-cone-expansion-symbols-2}
        For all~$n \in \Z$, the light cone expansion symbols $T^{(n)}(x, y)$ satisfy the relation
        \begin{align}
            \left( -\Box^{(1)} - m^2 \right) T^{(n)} = \left( n + \frac{\Box^{(1)} \sigma - d}{2} \right) T^{(n - 1)} \,.
        \end{align}
    \end{corollary}
    \begin{proof}
        The identity is obtained by direct computation using Lemma~\ref{lem:properties-of-light-cone-expansion-symbols}.
    \end{proof}

    The properties of the light cone expansion symbols and of the van Vleck-Morette determinant allow us to derive a
    recursive relation for the coefficients~$a_n(x,y)$ in the following ansatz,
    \begin{align}
        \label{eq:schwinger-dewitt-expansion}
        G(x, y) = \Delta^{\frac{1}{2}}(x, y) \sum_{n = 0}^{\infty} a_n(x, y) T^{(n)}(x, y) \,.
    \end{align}
    This formal series is called the \textbf{Schwinger-DeWitt expansion} of~$G(x, y)$.
    The recursive relations for the coefficients $a_n(x, y): E_y \to E_x$ are given by
    the transport equations
    \begin{align}
        &\sigma^{\mu} \nabla^{(1)}_{\mu} a_0(x, y) = 0 \,, \label{eq:a_0-recusrive-relation} \\
        &(n + \sigma^\mu \nabla^{(1)}_\mu) a_{n + 1}(x, y) \notag \\
        &\quad = \Delta^{-\frac{1}{2}}(x, y) \left( \Box^{(E, 1)} - V(x) \right) \left(\Delta^{\frac{1}{2}}(x, y) a_n(x, y)\right) \label{eq:a_n-recusrive-relation}\,.
    \end{align}
    The initial condition for $a_0(x, y)$ in~\eqref{eq:a_0-recusrive-relation} is given by $a_0(y, y) = \id_{E_y}$.
    The resulting initial-value problem simply describes
    parallel transport along the geodesic joining~$x$ and~$y$.
    We denote the unique solution by~$a_0(x, y) = U(x, y)$ and refer to~$U(x, y)$ as the \textit{parallel displacement operator}
    or in the case of the spin connection as \textit{spin parallel transport} map.

    \begin{lemma}
        The recursive relation~\eqref{eq:a_n-recusrive-relation} for the coefficients $a_n(x, y)$ is solved by the path-ordered integrals
        \begin{align}
            a_n(x, y) = U(x, y) \int_0^1 \dd s_n \cdots \int_0^{s_2} \dd s_1 ~ B(s_n) \cdots B(s_1) \id_{E_y} \, \label{eq:a_n-solution},
        \end{align}
        where $B: (0, 1) \to \Lin(\End(E_y))$ is an operator acting on $\End(E_y)$ defined along the geodesic
        $\gamma: [0, 1] \to M$ connecting $y = \gamma(0)$ and $x = \gamma(1)$ by
        \begin{align}
            B(s) &\defeq \Delta^{-\frac{1}{2}} \left(\gamma(s), y \right) U\left(y, \gamma(s) \right) \nonumber \\
            &\qquad \qquad \times \left( \Box^{(E, \gamma(s))} - V\left(\gamma(s)\right) \right)
            U\left( \gamma(s), y \right) \Delta^{\frac{1}{2}} \left( \gamma(s), y \right),
        \end{align}
        where $\Box^{(E, \gamma(s))}$ denotes the connection Laplacian at~$\gamma(s)$.
    \end{lemma}
    \begin{proof}
        Let $\gamma: [0, 1] \to M$ be the geodesic from $y = \gamma(0)$ to $x = \gamma(1)$, then substituting
        \begin{align}
            s^n a_n(\gamma(s), y) = U(\gamma(s), y) b_n(s)
        \end{align}
        with $b_n(s) \in \End(E_y)$ into~\eqref{eq:a_0-recusrive-relation} and \eqref{eq:a_n-recusrive-relation} gives
        \begin{align}
            b_0(s) = \id_{E_y} \,, \quad
            \dv{s} b_{n + 1}(s) = B_n(s) b_n(s) \,, \quad
            b_{n + 1}(0) &= 0\,.
        \end{align}
        The equation for $b_{n + 1}$ is solved by integration, i.e.
        \begin{align}
            b_{n + 1}(s) = \int_0^{s} \dd s' ~ B_n(s') b_n(s') \,.
        \end{align}
        Iteration then gives the path-ordered integrals.
    \end{proof}

    This completes the derivation of the Schwinger-DeWitt expansion for Klein-Gordon type operators.
    In the context of causal fermion systems for Minkowski spacetime, this expansion is often referred to as the
    \textit{light cone expansion}.
    In this setting, the Schwinger-DeWitt expansion for the Klein-Gordon equation without potential
    truncates after the first term and the fundamental solution is given by $G(x, y) = T^{(0)}(x, y)$.
    In the presence of a scalar potential $V(x)$, the formal series was first given in~\cite{firstorder} explicitly up
    to first order in $V(x)$.
    The following example shows that the Schwinger-DeWitt expansion indeed reproduces the same result.

    \begin{example}
        Let $(M, g) = (\R^4, \eta)$ be the four-dimensional Minkowski spacetime with standard metric $\eta = \diag(1, -1, -1, -1)$
        and $E = M \times \C^4$ the trivial spinor bundle with $\nabla_\mu = \partial_\mu$.
        The geodesic connecting $x = \gamma(1)$ and $y = \gamma(0)$ is given by $\gamma(s) = sx + (1 - s)y$,
        the van Vleck-Morette determinant is $\Delta(x, y) = 1$ and the spin parallel transport map is $\id_{\C^4}$.
        Restricting the Schwinger-DeWitt expansion~\eqref{eq:schwinger-dewitt-expansion} to first order in $V(x)$,
        we obtain the formal series
        \begin{align}
            G(x, y) &= T^{(0)}(x, y) \nonumber \\
            &\quad - \sum_{n = 0}^{\infty} \frac{1}{n!} \int_0^1 (s - s^2)^n (\Box^n V)(s x + (1 - s) y) \, \dd s \, T^{(n + 1)}(x, y) \nonumber \\
            &\quad + \mathcal{O}(V^2) \,.
        \end{align}
    \end{example}

    \subsection{Causal Fermion Systems and the Causal Action Principle}\label{subsec:causal-fermion-systems-and-the-causal-action-principle}

    The theory of causal fermion systems is a recent approach to fundamental physics.
    Unlike standard formulations of wave functions on a pre-existing geometric background, a causal fermion system is defined
    entirely on a set of operators acting on a Hilbert space $\H$.
    The dynamics are governed by a variational principle known as the causal action principle.
    We briefly recall the basic definitions and mathematical objects, following the standard references~\cite{cfs}.

    \begin{definition}[Causal Fermion System]
        \label{def:causal-fermion-system}
        Let $(\H, \braket{\cdot}{\cdot}_\H)$ be a separable complex Hilbert space.
        For a given parameter $n \in \N$ (the \textbf{spin dimension}), we let $\F$ denote the set of all self-adjoint
        operators $x \in \Lin(\H)$ of finite rank, which have at most $n$ positive and at most $n$ negative eigenvalues.
        A \textbf{causal fermion system} is the triple $(\H, \F, \rho)$, where $\rho$ is a positive Borel
        measure on $\F$.
    \end{definition}

    While $\H$ and $\F$ are kept fixed, the measure $\rho$ serves as the fundamental degree of freedom of the theory.
    As shown in Section~\ref{subsec:construction-of-the-causal-fermion-system}, the support of the measure $\rho$ corresponds
    to the classical manifold for the case of globally hyperbolic spacetimes.
    Therefore, we introduce the notion of a \textit{spacetime} for a causal fermion system $M \defeq \supp \rho \subset \F$.
    Similarly, the following notions are motivated by the classical correspondence for globally hyperbolic spacetimes.

    In contrast to a point in classical spacetime, a point $x \in M$ is an operator acting on $\H$.
    Thus, it also encodes local geometric information.
    The \textit{spinor space} $S_x M \subset \H$ is defined as the image of the operator $x$, i.e.~$S_x M \defeq x(\H)$.
    Since $x$ has rank at most $2n$, the spinor space $S_x M$ defines a finite-dimensional subspace of $\H$ of dimension
    at most $2n$.
    Note that not all spinor spaces necessarily have the same dimension.
    Nevertheless, we introduce the notation
    \begin{align}
        SM \defeq \disjointUnion_{x \in M} S_x M \,,
    \end{align}
    which in the case considered in the paper corresponds to the spinor bundle of the classical spacetime.

    The operators $x \in M$ directly equip every spinor space $S_x M$ with an indefinite inner product $\spinbraket{\cdot}{\cdot}_x$,
    defined via the Hilbert space scalar product by
    \begin{align}
        \label{eq:cfs-spin-inner-product}
        \spinbraket{u}{v}_x \defeq -\bra{u} x \ket{v} \qquad \text{for all } u, v \in S_x \,.
    \end{align}
    We refer to $\spinbraket{\cdot}{\cdot}_x$ as \textit{spin inner product} at $x$.

    To connect the global Hilbert space with the local spinor spaces, we introduce the \textit{wave evaluation operator}
    $\Psi(x): \H \to S_x M$, which orthogonally projects a Hilbert space vector onto the spinor space $S_x M$.
    The notion of the wave evaluation operator becomes directly clear when one considers the assignment
    \begin{align}
        \psi^u: M \to S M, \quad x \mapsto \Psi(x) u \,.
    \end{align}
    The so-called \textit{physical wave function} $\psi^u$ defines for each vector $u \in \H$ the equivalent of a section of
    the spinor bundle $SM$.

    Its adjoint $\Psi(x)^*: S_x M \to \H$, taken with respect to the Hilbert space scalar product and the spin inner product,
    acts as an inclusion map weighted by the operator $x$.
    A direct computation shows that
    \begin{align}
        \label{eq:x-psi-relation}
        \Psi(x)^* = -x|_{S_x M} \quad \text{ and } \quad x = - \Psi(x)^* \Psi(x) \,.
    \end{align}

    The dynamics of a causal fermion system are governed by the \textbf{causal action principle}.
    We define the Lagrangian $\L(x, y)$ in terms of the spectral weight of the operator product $xy$ for $x, y \in M$.
    The causal action $S[\rho]$ is then constructed by integrating the Lagrangian over all pairs of spacetime points:
    \begin{align}
        \label{eq:causal-action}
        S[\rho] = \int_{\F} \int_{\F} \L(x, y) ~ \dd \rho(x) ~ \dd \rho(y) \,.
    \end{align}
    By applying relation~\eqref{eq:x-psi-relation} we get that
    \begin{align}
        \label{eq:xy-isospectral-to-closed-chain}
        xy = \Psi(x)^* \Psi(x) \Psi(y)^* \Psi(y) \isospectralto \Psi(x) \Psi(y)^* \Psi(y) \Psi(x)^*
    \end{align}
    where $\isospectralto$ means isospectral, i.e.~the operators on both sides have the same non-vanishing eigenvalues.
    While a priori the operator product $xy$ is defined on all of $\H$, the operator on the right-hand side has finite
    rank and therefore allows for a simplified analysis.
    Hence, the central object of a causal fermion system is the \textit{kernel of the fermionic projector}
    $P(x, y): S_y M \to S_x M$, defined as
    \begin{align}
        \label{eq:kernel-of-the-fermionic-projector}
        P(x, y) = -\Psi(x) \Psi(y)^* \,.
    \end{align}
    In addition, we refer to the product $P(x, y) P(y, x)$ appearing in~\eqref{eq:xy-isospectral-to-closed-chain}
    as the \textit{closed chain}, denoted by $A_{xy}: S_x M \to S_x M$.
    Thus, we have the relation
    \begin{align}
        \label{eq:closed-chain-def}
        xy \isospectralto A_{xy} \defeq P(x, y) P(y, x) \,.
    \end{align}

    By integrating the kernel of the fermionic projector over the spacetime against the measure $\rho$,
    the \textit{fermionic projector} $P$ acting on a physical wave function $\psi^u$ is defined globally as
    \begin{align}
        \label{eq:fermionic-projector-def}
        (P \psi^u) (x) \defeq \int_\F P(x, y) \psi^u(y) \dd \rho(y) \,.
    \end{align}

    Although $x, y \in \F$ are self-adjoint operators, the product $xy$ is in general not self-adjoint.
    Therefore, the operator $xy$ has complex eigenvalues.
    Since, $x$ and $y$ are finite-ranked, $xy$ also has rank $k \leq 2n$.
    We denote the non-vanishing eigenvalues of $xy$ by $\lambda^{xy}_i$ with $i = 1, \dots k$.
    For notational simplicity, from now on we consider $2n$ eigenvalues of $xy$, where $\lambda^{xy}_1, \dots \lambda^{xy}_k \neq 0$
    are the non-vanishing eigenvalues of $xy$.
    For $k < 2n$, we then set $\lambda^{xy}_{k + 1} \dots \lambda^{xy}_{2n} = 0$.
    As stated above, the operator $xy$ and the closed chain $A_{xy}$ have the same non-vanishing eigenvalues.
    Thus, $\lambda^{xy}_1, \dots \lambda^{xy}_{2n}$ are also the eigenvalues of $A_{xy}$.
    Based on these eigenvalues, we define the Lagrangian of the causal action principle as
    \begin{align}
        \L(x, y) \defeq \frac{1}{4n} \sum_{i, j = 0}^{2n} \left( \abs{\lambda^{xy}_i} - \abs{\lambda^{xy}_j} \right)^2. \label{eq:lagrangian}
    \end{align}
    Then, the action is varied under the following constraints
    \begin{align}
        \text{Volume:} && \rho(\F) &= \text{const} \,,\\
        \text{Trace:} && \int_\F \tr_\H x ~ \dd \rho(x) &= \text{const} \,,\\
        \text{Boundedness:} && \iint_{\F \times \F} \left( \sum_{i = 1}^{2n} \abs{\lambda^{xy}_i} \right)^2 ~ \dd \rho(x) ~ \dd \rho(y) &< \infty \,.
    \end{align}

    This variational principle is mathematically well-posed if $\H$ is finite-dimensional.
    For the existence theory and the analysis of general properties of minimizing measures, we refer to~\cite{discrete, continuum}.
    In order to treat the infinite-dimensional case, one can consider a limiting case where
    both the dimension of~$\H$ and the total volume~$\rho(\F)$ tend to infinity
    (for a rigorous construction in this direction see~\cite[Section~7]{minmov}).
    The minimality condition for the causal action principle gives rise to the \ac{el} equations for a causal fermion system~\cite[Proposition 2.3]{nonlocal}.
    These \ac{el} equations are also well-defined in the infinite-dimensional setting.

    \begin{proposition}{\ac{el} equations}
        \label{prop:el-equation}
        Let $\rho$ be a minimizer of the causal action principle with Lagrange multipliers $\mathfrak{s}, \mathfrak{r}, \kappa$
        for the volume, trace, and boundedness constraints.
        Then, there exist $\mathfrak{s}, \mathfrak{r}, \kappa > 0$, such that the function
        \begin{align}
            \ell(x) \defeq \int_{\F} \L(x, y) + \kappa \left( \sum_{i = 1}^{2n} \abs{\lambda^{xy}_i} \right)^2 \dd \rho(y) - \mathfrak{r} \tr_\H x - \mathfrak{s}
        \end{align}
        is minimal and vanishes on spacetime $M = \supp \rho$, i.e.
        \begin{align*}
            \ell|_M = \inf_{x \in \F} \ell(x) = 0 \,.
        \end{align*}
    \end{proposition}

    Although the boundedness constraint is necessary for the general existence theory of minimizing measures~\cite{discrete, continuum},
    the associated Lagrange multiplier $\kappa$ has been shown to be quantitatively extremely small~\cite[Appendix 3]{jacobson}.
    More crucially for the analysis in this paper, its contribution to the restricted Euler-Lagrange equations
    compared to the geometrical contributions is negligible.
    Therefore, when deriving the macroscopic spacetime dynamics, we set $\kappa = 0$ without affecting the resulting
    geometric field equations.

    As it was shown in~\cite{jacobson}[Appendix 3], the value of the parameter $\kappa$ is very small.
    Throughout this paper, we set $\kappa = 0$ and neglect the term involving the boundedness constraint.

    A direct consequence of the \ac{el} equations is that for a variation $\delta x$ of a spacetime point $x \in M$, which
    can be expressed using relation~\eqref{eq:x-psi-relation} as
    \begin{align}
        \delta x = - \delta \Psi(x)^* \Psi(x) - \Psi(x)^* \delta \Psi(x) \,,
    \end{align}
    the corresponding variation of the function $\ell(x)$ has to vanish.
    By definition of the kernel of the fermionic projector~\eqref{eq:kernel-of-the-fermionic-projector} and
    the closed chain~\eqref{eq:closed-chain-def}, we directly get
    \begin{align}
        \delta P(x, y) &= -\delta \Psi(x) \Psi(y)^* - \Psi(x) \delta \Psi(y)^* \,,\\
        \delta A_{xy} &= \delta P(x, y) P(y, x) + P(x, y) \delta P(y, x) \,.
    \end{align}
    Thus, we have
    \begin{align}
        0 &\stackrel{!}{=} \delta \ell(x) = \int_{\F} \delta^{(1)} \L(x, y) \dd \rho(y) - \mathfrak{r} \tr_\H \left[\delta x\right] \nonumber \\
        &= \int_{\F} \tr_{S_x M} \left[\pdv{\L(x, y)}{A_{xy}} \delta^{(1)} A_{xy} \right]~ \dd \rho(y) + \mathfrak{r} \Re \tr_\H \left[\delta \Psi^*(x) \Psi(x) \right] \nonumber \\
        &= \int_{\F} 2 \Re \tr_{S_x M}  \left[\pdv{\L(x, y)}{A_{xy}} P(x, y) \delta^{(2)} P(y, x) \right] ~ \dd \rho(y) \nonumber \\
        &\quad + 2 \mathfrak{r} \Re \tr_{S_x M} \left[\delta \Psi(x) \Psi^*(x) \right] \nonumber \\
        &= - 2 \Re \tr_{S_x M} \left[\left( \int_{\F} \pdv{\L(x, y)}{A_{xy}} P(x, y) \Psi(y) ~ \dd \rho(y) - \mathfrak{r} \Psi(x) \right) \delta \Psi^*(x) \right] \,.
    \end{align}
    Since $\delta \Psi^*(x)$ is arbitrary, we conclude that for a minimizing measure $\rho$ the \textit{restricted \ac{el} equations}
    \begin{align}
        \label{eq:restricted-euler-largrange-equations}
        (Q \Psi)(x) = \int_{\F} Q(x, y) \Psi(y) ~ \dd \rho(y) = \mathfrak{r} \Psi(x)
    \end{align}
    with $Q(x, y) \defeq \pdv{\L(x, y)}{A_{xy}} P(x, y)$ must hold for all $x \in M$.
    Note that this is only a necessary condition, as it only considers variations of points in $M$, the minimality condition
    on all of $\F$ is stronger.
    We refer to a measure which satisfies~\eqref{eq:restricted-euler-largrange-equations} for some $\mathfrak{r} \geq 0$
    as a \textit{critical point} of the causal action.

    Similar to the derivation of the restricted \ac{el} equations~\eqref{eq:restricted-euler-largrange-equations}, we
    consider a second variation of the wave evaluation operator $\Psi$.
    Then, preserving the restricted \ac{el} equations means that
    \begin{align}
        \label{eq:linearized-field-equations}
        (\delta Q \Psi)(x) + (Q \delta \Psi)(x) - \mathfrak{r} \delta \Psi(x) = 0 \,.
    \end{align}
    These equations are referred to as the \textit{linearized field equations}.
    Suppose $\rho$ is a critical point of the causal action, then the linearized field equations describe the dynamics
    of small perturbations $\delta \Psi$.
    Thus, the linearized field equations tell us which variations are allowed such that the perturbed measure is still a
    critical point of the causal action.
    For a detailed analysis of the linearized field equations, we refer to~\cite{nonlocal}.

    \section{The $i\varepsilon$-Regularization of the Schwinger-DeWitt Expansion}\label{sec:iepsilon-regularization-of-the-schwinger-dewitt-expansion}

    When analyzing the singular structure of the Schwinger-DeWitt expansion~\eqref{eq:schwinger-dewitt-expansion}, more
    precisely the distributions $T^{(n)}(x, y)$, one finds that they are singular on the light cone, i.e.~for $\sigma(x, y) = 0$.
    In this section, we shall develop a systematic method for regularizing these singularities such that the resulting
    regularized expansion is still a bi-solution of the Klein-Gordon equation~\eqref{eq:spin-klein-gordon}.
    A naive approach is to choose a Cauchy time function $T: M \to \R$ and shift
    \begin{align}
        \sigma(x, y) \to \sigma(x, y) - i \varepsilon (T(y) - T(x)) \,.
    \end{align}
    However, as shown in~\cite{reghadamard}, this regularization only gives a bi-solution at zero-th order in $\varepsilon$.
    A more sophisticated approach is to use the ansatz
    \begin{align}
        \label{eq:regularized-geodesic-interval}
        \sigma(x, y) \to \sigma^\varepsilon(x, y) \defeq \sigma(x, y) - i \varepsilon f(x, y)
    \end{align}
    for a suitable function $f: U \times U \to \C$.
    In the following, we derive the precise definition of what is meant by \textit{suitable} in this case.
    The advantage of this ansatz is its chart-independence, providing a regularization scheme that does not rely on a
    specific choice of coordinates or time direction.

    A systematic procedure to derive transport equations for $f(x, y)$ was derived in~\cite{reghadamard}.
    However, in this approach, a real-valued function $f$ was considered, which only gives an approximate regularized bi-solution
    when used in the Schwinger-DeWitt expansion.
    To allow for an exact regularized bi-solution, we require that the fundamental bi-tensor identity~\eqref{eq:fundamental-bi-tensor-identity}
    is preserved under the replacement $\sigma \to \sigma^\varepsilon$.
    This gives the following system of partial differential equations
    \begin{align}
        f &= g_x(\nabla^{(1)} \sigma, \nabla^{(1)} f) - \frac{i \varepsilon}{2} g_x(\nabla^{(1)} f, \nabla^{(1)} f) \label{eq:regularization-function-condition-x} \,,\\
        &= g_y(\nabla^{(2)} \sigma, \nabla^{(2)} f) - \frac{i \varepsilon}{2} g_y(\nabla^{(2)} f, \nabla^{(2)} f) \label{eq:regularization-function-condition-y} \,.
    \end{align}
    At zero-th order in $\varepsilon$, i.e.
    \begin{align}
        g_x(\nabla^{(1)} \sigma, \nabla^{(1)} f) = f = g_y(\nabla^{(2)} \sigma, \nabla^{(2)} f) \,,
    \end{align}
    this equation coincides with the result derived in~\cite{reghadamard}.

    It is a classical transport equation along geodesics, which, when viewed as an ordinary differential equation,
    possesses a unique solution for a prescribed initial value.
    To systematically assign these initial values across the spacetime, we specify data on a chosen Cauchy hypersurface
    and transport it along the geodesics emanating from it.
    Specifically, let $\Sigma$ be a Cauchy surface and let $\alpha_x, \beta_x: T_x M \setminus \{0\} \to \C$ be smooth functions.
    Further, we require that the assignments $\alpha: x \mapsto \alpha_x$ and $\beta: x \mapsto \beta_x$ vary smoothly on $\Sigma$.
    Suitable initial conditions for~\eqref{eq:regularization-function-condition-x} are given by
    \begin{align}
        \lim_{s \searrow 0} f\left(\exp_x(s \xi), x\right) &= \alpha_x(\xi) & (D^{(1)}_\xi f)(x, x) &= \beta_x(\xi)
    \end{align}
    for all $x \in \Sigma$ and $\xi \in T_x M \setminus \{0\}$.
    In other words, we require that in each direction $\xi$, the function $f$ is continuous along the geodesic and its
    directional derivative exists.
    Note that this does not imply that $f$ is continuous in one of its arguments or that the total derivative exists.

    The higher orders in $\varepsilon$ are derived via perturbative analysis.
    To this end, we represent $f$ as the formal series
    \begin{align}
        f(x, y) = \sum_{n = 0}^{\infty} f^{(n)}(x, y) \, \varepsilon^n \,.
    \end{align}
    At each order in $\varepsilon$, we obtain
    \begin{align}
        g_x\left(\nabla^{(1)} \sigma, \nabla^{(1)} f^{(0)}\right) - f^{(0)} &= 0 \,, \\
        g_x\left(\nabla^{(1)} \sigma, \nabla^{(1)} f^{(n + 1)}\right) - f^{(n + 1)} &= \frac{i}{2} \sum_{k = 0}^n g_x\left(\nabla^{(1)} f^{(k)}, \nabla^{(1)} f^{(n - k)}\right) \label{eq:higher-order-transport}\,
    \end{align}
    with analogous equations holding for the second argument.
    These are transport equations along the geodesic $\gamma$, where the parameterization is fixed by the initial conditions.
    Thus, setting $f^{(n)}_{\gamma}(a, b) = f^{(n)}(\gamma(a), \gamma(b))$ gives~\cite[Appendix B]{dgc}
    \begin{align}
        g_x\left(\nabla^{(1)} \sigma, \nabla^{(1)} f^{(n)}\right) &= (a - b) \left. \pdv{s} \right|_{s = a} f_\gamma(s, b) \,, \\
        g_y\left(\nabla^{(2)} \sigma, \nabla^{(2)} f^{(n)}\right) &= (b - a) \left. \pdv{s} \right|_{s = b} f_\gamma(a, s) \,.
    \end{align}
    The transport equations reduce to ordinary differential equations, which have the following solution.

    \begin{lemma}
        \label{lem:ode-unique-solution}
        Let $a < b \in \R$, $s_0 \in (a, b)$, $c \in \C$ and $g: [a, b] \to \C$ continuous.
        Then for $s \in (a, b)$ the ordinary differential equation
        \begin{align}
            (s - a)
            f'(s) - f(s) &= g(s)\,,
            & f(s_0) &= (s_0 - a) c
        \end{align}
        has the unique differentiable solution $f: (a, b) \to \C$ given by
        \begin{align}
            f(s) = (s - a) \left( c + \int_{s_0}^s \frac{g(t)}{(t - a)^2} \, \dd t \right) \,.
        \end{align}
    \end{lemma}
    \begin{proof}
        For $s \in (a, b)$ substitute $h(s) = \frac{f(s)}{s - a}$, then
        \begin{align*}
            h'(s) = \frac{(s - a) h'(s) - h(s)}{(s - a)^2} = \frac{g(s)}{(s - a)^2} \,.
        \end{align*}
        This equation is solved by integration, giving the solution $f_a$ defined on $(a, b)$ with integration constant $c \in \C$.
    \end{proof}

    \begin{example}
        \label{eg:iepsilon-regularization-minkowski}
        Let $(M, g) = (\R^4, \eta)$ be the four-dimensional Minkowski spacetime with cartesian metric $\eta = \diag(1, -1, -1, -1)$.
        Consider the Cauchy slice $\Sigma = \{0\} \times \R^3$ and the initial conditions
        \begin{align}
            \alpha_x(\xi) = -\frac{i}{2} \varepsilon \quad \text{ and } \quad \beta_x(\xi) = \xi^0
        \end{align}
        for $x \in \Sigma$.
        Then, the application of Lemma~\ref{lem:ode-unique-solution} gives
        \begin{align}
            f_\gamma^{(0)}(0, s) = s \dot{\gamma}^0(0) \quad \Rightarrow \quad f^{(0)}(x, y) = y^0 - x^0 \,.
        \end{align}
        The next order is then given by
        \begin{align}
            f_\gamma^{(1)}(0, s) = \frac{i}{2} \quad \Rightarrow \quad f^{(1)}(x, y) = \frac{i}{2} \,.
        \end{align}
        Thus, the formal series for $f(x, y)$ truncates after the first order, and we get
        \begin{align}
            f(x, y) = y^0 - x^0 + \frac{i}{2} \varepsilon \,.
        \end{align}
    \end{example}

    The construction presented above only works for $x \neq y$, i.e.~there exists a unique geodesic.
    The value of $f_\gamma(x, x)$ might depend on the geodesic, hence $f$ is not necessarily defined on the diagonal.
    For the present analysis, it suffices that $f$ is defined almost everywhere.
    Further, we require that $f$ satisfies $f(x, y) = - \complexconj{f(y, x)}$ and that the gradient w.r.t.~its first
    argument of the real part of $f$ is a past-directed timelike vector field in $T_x M$.
    Therefore, we introduce the following notion.

    \begin{definition}
        \label{def:regularizing-scalar-field}
        A \textbf{regularizing scalar field}, is a solution $f: U \times U \to \C$ of the PDE system of~\eqref{eq:regularization-function-condition-x}
        and~\eqref{eq:regularization-function-condition-y} satisfying, for all $x, y \in U$ with $x \neq y$, the properties
        \begin{enumerate}[label=\roman*), font=\upshape]
            \item $\displaystyle f(x, y) = - \complexconj{f(y, x)}$, \\
            \item $\displaystyle f(x, y) \neq 0$,  \\
            \item $\displaystyle (\nabla^{(1)} \Re f)(x, y)$ is a past directed timelike vector field, \\
            \item{
                for all $\xi \in T_x M \setminus \{0\}$
                \begin{align}
                    \displaystyle\lim_{s \searrow 0} \Im f(\exp_x(s \xi), x) > 0 \quad \text{ and is order } \mathcal{O}(\varepsilon) \,.
                \end{align}
            }
        \end{enumerate}
    \end{definition}

    For a regularizing scalar field $f: U \times U \to \C$, the regularized Synge's world function defined by~\eqref{eq:regularized-geodesic-interval}
    satisfies for all $x, y \in U$
    \begin{align}
        \sigma^\varepsilon(x, y) &\neq 0 \,, & \sigma^{\varepsilon}(x, y) &= \complexconj{\sigma^\varepsilon(y, x)} \,, & g_x(\xi, \xi) &= 2 \sigma^{\varepsilon}(x, y) \,,
    \end{align}
    where\footnote{
        The minus sign is introduced to match the convention in Minkowski spacetime, where $\xi^\mu = \left(y^0 - x^0 - i \varepsilon, \vec{y} - \vec{x}\right)$
        for $f(x, y) = y^0 - x^0 - \frac{i \varepsilon}{2}$ (see Example~\ref{eg:iepsilon-regularization-minkowski}).
    } $\xi \defeq - \nabla^{(1)} \sigma^{\varepsilon}$ .
    Requiring that the regularized Synge's world function satisfies the fundamental bi-tensor identity allows us to simply replace
    $\sigma$ by $\sigma^\varepsilon$ in the Schwinger-DeWitt Expansion~\eqref{eq:schwinger-dewitt-expansion} and obtain
    a fully regularized bi-solution to the Klein-Gordon equation~\eqref{eq:spin-klein-gordon}.

    At this point we want to remark that the replacement $\sigma \to \sigma \pm i \varepsilon f$ in the Schwinger-DeWitt Expansion~\eqref{eq:schwinger-dewitt-expansion}
    gives two bi-solutions.
    In the presence of a global timelike Killing field, they are commonly referred to as the positive and negative frequency
    Wightman bi-solutions.
    In general globally hyperbolic spacetimes, the notion of positive and negative frequencies is not meaningful.

    For this work, we do not need to specify the regularizing scalar field directly.
    Thus, the sign in the replacement does not play an important role.
    To align the results with the known analysis of causal fermion systems in Minkowski spacetime~\cite{cfs},
    we consider the replacement $\sigma \to \sigma - i \varepsilon f$.

    \section{Causal Fermion Systems for Globally Hyperbolic Spacetimes}\label{sec:causal-fermion-systems-for-globally-hyperbolic-spacetimes}

    In this section, we construct a causal fermion system (Definition~\ref{def:causal-fermion-system}) which corresponds
    to a globally hyperbolic spacetime $(M, g)$.
    This correspondence is established by studying a subset of the solution space of the Dirac Equation~\eqref{eq:dirac-equation}.
    More precisely, representing the vectors of the Hilbert space $\H$ of the causal fermion system by certain solutions
    of the Dirac equation allows us to equip the support of the measure $\rho$ with a smooth structure.

    This identification also allows us to locally represent the kernel of the fermionic projector using the regularized
    Schwinger-DeWitt expansion introduced in Section~\ref{subsec:the-schwinger-dewitt-expansion}.
    The regularization with parameter $\varepsilon$ is a direct consequence of the fact that the kernel of the fermionic
    projector, as defined by~\eqref{eq:kernel-of-the-fermionic-projector}, is non-singular for all $x, y \in \supp \rho$.
    This defines a family of causal fermion systems labeled by the regularization parameter $\varepsilon > 0$.
    We refer to the \textit{continuum limit} as the limit $\varepsilon \searrow 0$.
    This limit has been studied in detail for Minkowski spacetime, e.g.~see~\cite{cfs-current} or textbook~\cite{cfs}.
    Due to the considerations in the previous sections, the analysis directly generalizes to arbitrary globally hyperbolic spacetimes,
    because the kernel of the fermionic projector has locally the same singular structure as in Minkowski spacetime.

    As $\varepsilon \searrow 0$, the light cone expansion symbols $T^{(n)}(x, y)$ (Definition~\ref{def:light-cone-expansion-symbols})
    become singular on the light cone $\sigma(x, y) = 0$.
    The leading singularity is given by
    \begin{align}
        T^{(n)}(x, y) \propto \frac{(-2m^2)^{\nu_n}}{z^{2\nu_n}} = \frac{1}{\sigma^{\nu_n}(x, y)} \quad \text{ with } \quad \nu_n = \frac{d}{2} - 1 - n\,,
    \end{align}
    and the associated degree is $\deg T^{(n)} = \nu_n$.
    Hence, the continuum limit analysis is given by an expansion in the degrees of $T^{(n)}$.
    By Lemma~\ref{lem:properties-of-light-cone-expansion-symbols}, multiplication of $T^{(n)}$ with $\sigma$ or $\sigma^\mu$
    reduces the degree by one or $\tfrac{1}{2}$ respectively.
    Applying covariant expansion of geometric objects as van Vleck-Morette determinant allows to determine the reduction
    of the degree for each term in the expansion.

    In Section~\ref{subsec:closed-chain-and-its-eigenvalues}, we begin by considering only the leading degree contributions.
    Afterward, in Section~\ref{sec:perturbative-analysis-of-the-euler-lagrange-equations}, we consider the next-to-leading
    contribution, which then gives rise to the Einstein field equations.

    First, we recall from Section~\ref{subsec:causal-fermion-systems-and-the-causal-action-principle} that a causal fermion
    system is defined by a Hilbert space $\H$ and a measure $\rho$ on a set of linear operators $\F \subset \Lin(\H)$.
    In the following, we first find a suitable choice for $\H$ by solutions to the Dirac Equation~\eqref{eq:dirac-equation}
    on $(M, g)$.

    \subsection{The Dirac Equation in Curved Spacetimes (Part II)}\label{subsec:the-dirac-equation-in-curved-spacetimes-2}

    In Section~\ref{subsec:the-dirac-equation-in-curved-spacetimes-1}, we already introduced the Dirac Equation~\eqref{eq:dirac-equation}.
    Let $(\Sigma_t)_{t \in \R}$ be a foliation of $M$ into spatial hypersurfaces.
    Given some initial condition $u_0 \in \Gamma(\Sigma_0, SM)$, then the Dirac equation for a spinor field $u \in \Gamma(M, SM)$
    with the requirement $u|_{\Sigma_0} = u_0$ defines a well-posed Cauchy problem with a unique solution~\cite[Theorem 3.2.11]{baer+ginoux}.
    In particular, if $u_0 \in \Gamma_{c}(\Sigma_0, SM)$, i.e.~for compactly supported $u_0$ on $\Sigma_0$, the corresponding
    solution $u$ is compactly supported on $\Sigma_t$ for every $t \in \R$.
    We denote the set of spinor fields that are compact on every spatial hypersurface by $\Gamma_{sc}(M, SM)$.
    On the set of \textit{compactly supported} spinor fields $\Gamma_{sc}(M, SM)$, we introduce the scalar product
    \begin{align}
        \label{eq:dirac-scalar-product}
        \left( u | v \right) = \int_{\Sigma_0} \spinbraket{u}{\gamma(\nu) v}_{(0, x)} \, \dd \mu_{\Sigma_0}(x),
    \end{align}
    where $\nu$ is the future directed normal vector field of $\Sigma_0$, $\mu_{\Sigma_0}$ denotes the volume form of the manifold $\Sigma_0$
    and $\spinbraket{\cdot}{\cdot}_x$ the spin inner product of the fiber $S_x M$.
    The completion of the spatially compact supported solutions of the Dirac equation with respect to the scalar product
    $\left( \cdot | \cdot \right)$ gives a Hilbert space $\H_m$.
    The index $m$ denotes the mass parameter of the Dirac Equation~\eqref{eq:dirac-equation}.

    Further, let $s_m^{\lor}, s_m^{\land}: \Gamma_c(M, SM) \to \Gamma_{sc}(M, SM)$ denote the advanced and retarded Green's operators of the Dirac
    equation.
    Then, the \textit{causal fundamental solution} is defined as the difference
    \begin{align}
        k_m \defeq \frac{1}{2 \pi i } \left( s_m^{\lor} - s_m^{\land} \right) : \Gamma_c(M, SM) \to \Gamma_{sc}(M, SM) \cap \H_m \,.
    \end{align}
    This operator plays an important role in the `quantization' of the theory, as we see in the following.
    In the language of algebraic quantum field theory, we have the following result~\cite[Theorem 5.2.5]{bratteli-robinson-1997}.

    \begin{theorem}
        Up to $*$-isomorphism, there exists a unique complex $C^*$ algebra $\mathcal{A}$ generated by the identity and
        the so-called smeared fields $\Psi: \Gamma_c(M, SM) \to \mathcal{A}$ satisfying for all $f, g \in \Gamma_c(M, SM)$
        \begin{enumerate}[label=(\arabic*), font=\upshape]
            \item $\Psi$ is antilinear (or equivalently $f \mapsto \Psi(f)^*$ is linear), \\[-0.25cm]
            \item $\displaystyle \Psi\big((i \gamma^\mu \nabla_\mu - m) f\big) = 0$, \\[-0.25cm]
            \item $\displaystyle \left\{\Psi(f), \Psi(g) \right\} = 0$, \\[-0.25cm]
            \item $\displaystyle \left\{\Psi(f), \Psi(g)^* \right\} = (k_m f | k_m g)$.
        \end{enumerate}
    \end{theorem}

    Properties $(3)$ and $(4)$ are referred to as the \textit{canonical anti-commutation relations} (CARs) and consequently
    the algebra $\mathcal{A}$ is referred to as the CAR algebra over the Hilbert space $\H_m$.
    On the algebra we introduce the notation of a \textit{state}.
    A state represents the physical condition or configuration of a quantum system.
    In the context of the CAR algebra $\mathcal{A}$, a state $\omega$ is a linear functional $\omega: \mathcal{A} \to \C$
    that is both normalized ($\omega(1) = 1$) and positive ($\omega(a^* a) \ge 0$ for all $a \in \mathcal{A}$).
    A state is \textit{pure} if
    \begin{align}
        \omega = \lambda \omega_1 + (1 - \lambda) \omega_2 \quad \text{ for } \lambda \in (0, 1) \quad \Longrightarrow \quad \omega = \omega_1 = \omega_2 \,,
    \end{align}
    i.e.~if $\omega$ cannot be represented as a convex combination of two different states $\omega_1$ and $\omega_2$.

    Further, a state is said to be \textit{quasi-free} if all its higher-order correlation functions can be expressed
    in terms of its two-point functions.
    Mathematically, a state $\omega$ on the CAR algebra is quasi-free if it vanishes on all monomials of odd degree, and
    its values on even-degree monomials are given by the pairwise contractions~\cite[Definition 17.26]{derzinski-gerard-quantum}
    \begin{align}
        \omega \left( a_1 \cdots a_{2n} \right) = \sum_{\pi \in P_n} \sgn \pi \prod_{i = 1}^n \omega \left( a_{\pi(2i - 1)} a_{\pi(2i)} \right) \,.
    \end{align}
    The summation in the above expression iterates over all disjoint pairings, i.e.~all permutations $\pi \in S_{2n}$
    with $\pi(2i - 1) < \pi(2i)$ and $\pi(1) < \pi(3) < \cdots < \pi(2n - 1)$.
    In more common terms for physicists, a quasi-free state obeys the \textit{Wick's Theorem} expansion.
    Quasi-free states represent the algebraic generalization of non-interacting quantum fields, this includes the vacuum
    and thermal equilibrium states of free fermions.

    Moreover, we restrict our attention to states which are \textit{particle-number preserving}, i.e a state which satisfies
    the relation
    \begin{align}
        \omega\left( \Psi(f)^* \, \Psi(g)^* \right) = 0 = \omega\left( \Psi(f) \, \Psi(g) \right)
    \end{align}
    for all $f, g \in \Gamma_c(M, SM)$.
    In other words, all two-point expectations involving two creation or two annihilation operators vanish.
    In the literature, this property is also referred to as a \textit{gauge-invariant state} (see~\cite[Proposition~17.32]{derzinski-gerard-quantum}).

    \begin{proposition}
        Let $\mathcal{A}$ be a CAR algebra and $\omega: \mathcal{A} \to \C$ be a quasi-free state, then there exists a unique,
        bounded linear operator $A$ such that
        \begin{align}
            \omega \left( \Psi(f) \, \Psi^*(g) \right) = \left( k_m f | A (k_m g) \right)
        \end{align}
        for all $f, g \in \Gamma_c (M, SM)$.
        Further, $A$ has the following properties:
        \begin{enumerate}[label=(\arabic*), font=\upshape]
            \item $A$ is self-adjoint and $0 \leq A \leq 1$, \\[-0.25cm]
            \item $\omega$ is a pure state if and only if $A$ is an orthogonal projection.
        \end{enumerate}
    \end{proposition}
    \begin{proof}
        \textit{(existence and uniqueness):} On $\mathcal{A}$, the unique $C^*$-norm satisfies
        \begin{align*}
            \norm{\Psi(f) \, \Psi(f)^*}_{\mathcal{A}} = \norm{\Psi(f)}_{\mathcal{A}}^2 = (k_m f | k_m f) = \norm{k_m f}^2_{\H_m} \,.
        \end{align*}
        Further, the state $\omega$ is bounded in the sense that
        \begin{align*}
            \abs{\omega\left(\Psi(f) \, \Psi(g)^*\right)}^2 &\leq \omega \left( \Psi(f) \, \Psi(f)^* \right) \omega \left( \Psi(g) \, \Psi(g)^* \right) \\
            &\leq \norm{\Psi(f) \, \Psi(f)^*}_{\mathcal{A}} \norm{\Psi(g) \, \Psi(g)^*}_{\mathcal{A}} = \norm{k_m f}^2_{\H_m} \norm{k_m g}^2_{\H_m}
        \end{align*}
        Consequently, the state also satisfies $\omega\left(\Psi(f) \Psi(f)^*\right) \leq \norm{k_m f}^2_{\H_m}$.
        Since any smooth, spatially compact solution to the Dirac equation has the representation $k_m f$ for $f \in \Gamma_c(M, SM)$,
        we have that
        \begin{align*}
            (k_m f | k_m g) \mapsto \omega\left(\Psi(f) \, \Psi(g)^*\right)
        \end{align*}
        is a densely defined, bounded, positive sesquilinear form on $\H_m$.
        By the Riesz representation theorem, there exists a unique, bounded linear operator $A$ such that
        \begin{align*}
            (k_m f | A | k_m g) = \omega\left(\Psi(f) \, \Psi(g)^*\right) \quad  \text{ for all } f, g \in \Gamma_c (M, SM)\,.
        \end{align*}

        \noindent
        \textit{(1):}
        Since $\omega\left(\Psi(f) \, \Psi(f)^* \right) \geq 0$, it directly follows that $A \geq 0$.
        The fact that $\omega\left(\Psi(f) \Psi(f)^*\right) \leq \norm{kf}^2_{\H_m}$ implies $A \leq 1$.
        By construction, $A$ is symmetric and therefore, as a bounded operator, also self-adjoint.

        \noindent
        \textit{(2):}~\cite{hugenholtz-kadison-1975}[Proposition 2.2] Suppose $A^2 \neq A$, then by the spectral theorem
        there exists a projection $\Pi$ and a $t > 0$, such that $0 \leq A \pm t \Pi \leq 1$.
        Hence, we can write
        \begin{align*}
            \omega\left(\Psi(f) \, \Psi(g)^*\right) &= \frac{1}{2} (k_m f | A + t \Pi | k_m g) + \frac{1}{2} (k_m f | A - t \Pi | k_m g) \\
            &= \frac{1}{2} \omega_+\left(\Psi(f) \, \Psi(g)^*\right) + \frac{1}{2} \omega_-\left(\Psi(f) \, \Psi(g)^*\right) \,.
        \end{align*}
        Thus, $\omega$ is a convex combination of the states $\omega_+$ and $\omega_-$, i.e. $\omega$ is not a pure state.
    \end{proof}

    We refer to $A$ as the \textit{one-particle density operator}.
    Because $A$ is positive and self-adjoint, we can define its inverse square root $A^{- \frac{1}{2}}$
    via the spectral calculus as
    \begin{align}
        A^{- \frac{1}{2}} = \int_{\sigma(A)} \lambda^{- \frac{1}{2}} \dd E_{\lambda} \,,
    \end{align}
    where $E$ denotes the spectral measure of $A$.
    Its domain is given by
    \begin{align}
        \mathcal{D}(A^{- \frac{1}{2}}) = \left\{ u \in \overline{A(\H_m)} \:\middle|\: \int_{\sigma(A)} \lambda^{- \frac{1}{2}} \dd \left( u | E_\lambda u \right) < \infty  \right\} \,,
    \end{align}
    which lies dense in $\overline{A(\H_m)}$.
    By completing this domain with respect to the newly defined scalar product
    \begin{align}
        \label{eq:fundamental-scalar-product}
        \braket{u}{v} \defeq \left( A^{- \frac{1}{2}} u \, \middle| A^{- \frac{1}{2}} v \right)
    \end{align}
    we formally construct the Hilbert space $\H$ of the causal fermion system.

    From the perspective of algebraic quantum field theory in curved spacetimes, the one-particle density operator $A$
    acts as the covariance operator that uniquely determines the two-point Wightman function of the quasi-free state $\omega$.
    In the framework of causal fermion systems, however, $A$ has a direct physical interpretation: it specifies the occupied
    fermionic modes.
    These occupied modes collectively encode the spacetime geometry and the bosonic fields (for more details, see~\cite{cfs}).
    This establishes a direct connection between the algebraic quasi-free state $\omega$ and the fundamental Hilbert space $\H$.
    Specifically, $\H$ is built exactly from the scalar product defined by $A$, meaning it serves as the rigorous one-particle
    completion of these occupied background modes.
    Consequently, the fundamental Hilbert space provides the concrete mathematical foundation that fully characterizes
    the quasi-free state $\omega$.

    \begin{example}
        \label{eg:minkowski-hilbert-space}
        For the Minkowski spacetime $(M, g) = (\R^4, \eta)$, the causal fundamental solution $k_m$ is translation-invariant
        and, in momentum space, is supported on the mass shell $p^2 = m^2$.
        The mass shell naturally decomposes into the upper mass shell (positive frequencies) and the lower mass shell
        (negative frequencies).

        In algebraic quantum field theory, the Minkowski vacuum state $\omega_0$ is uniquely selected by requiring
        translation invariance and the \textit{spectrum condition} (positivity of the physical energy)~\cite{waldQFT}.
        This condition algebraically restricts the support of the two-point function to a single frequency sector.
        Without loss of generality, the sector is chosen such that the one-particle density operator $A$ coincides with
        the spectral projection operator onto the negative-frequency solutions (the lower mass shell).

        Because $A$ is a projection operator, we conclude that $\omega_0$ is pure and that
        $A$ as well as the operator $A^{-\frac{1}{2}}$ act as the identity on
        \begin{align*}
            \overline{A(\H_m)} = \overline{\left\{ u \text{ negative energy solution of the Dirac equation }\right\}}
        \end{align*}
        Thus, the fundamental Hilbert space $\H$ is then given by $\overline{A(\H_m)}$ and the scalar product
        $\braket{\cdot}{\cdot}$ (defined by~\eqref{eq:fundamental-scalar-product})
        is just the standard scalar product $(\cdot | \cdot)$ (defined by~\eqref{eq:dirac-scalar-product}) restricted
        to $\H$.
    \end{example}

    \subsection{Construction of the Causal Fermion System}\label{subsec:construction-of-the-causal-fermion-system}

    The physically relevant states in this context are typically required to satisfy the Hadamard condition to ensure a
    sensible short-distance singularity structure.
    On the other hand, the kernel of the fermionic projector as defined in~\eqref{eq:kernel-of-the-fermionic-projector}
    is regular for all $x$ and $y$.
    Thus, a rigorous identification of the objects defined so far requires a regularization of singularities.
    This is achieved by introducing a regularization parameter $\varepsilon$ which handles these UV divergences, thereby
    allowing for a mathematically well-defined state that mimics this preferred singular structure in the continuum limit $\varepsilon \searrow 0$.
    Formally, the regularization is handled by a \textit{regularization operator} defined as follows~\cite{finite}:

    \begin{definition}[Regularization operator]
        \label{def:regularization-operator}
        A family $(R_\varepsilon)_{\varepsilon>0}$ of bounded linear operators on $\H$ are called
        \textbf{regularization operators} if they have the following properties:

        \begin{enumerate}[label=(\roman*), font=\upshape]
            \item{
                Vectors of the Hilbert space are mapped to smooth solutions:
                \begin{align}
                    \mathcal{R}_\varepsilon : \H \to \Gamma(M, S M) \cap \H_m \,.
                \end{align}
            }
            \item{
                For every $\varepsilon > 0$ and $x \in M$, there is a constant $c > 0$ such that
                \begin{align}
                    \norm{(\mathcal{R}_\varepsilon u)(x)} \le c \norm{u} \quad \text{ for all } u \in \H \,.
                \end{align}
                (where the norm on the left is any norm on $S_x M$).
            }
            \item{
                In the limit $\varepsilon \searrow 0$, the regularization operators go over to the identity with strong
                convergence of $\mathcal{R}_\varepsilon$ and $\mathcal{R}_\varepsilon^*$, i.e.
                \begin{align}
                    \mathcal{R}_\varepsilon u, \mathcal{R}_\varepsilon^* u \xrightarrow{\varepsilon \searrow 0} u \text{ in } \H \quad \text{ for all } u \in \H \,.
                \end{align}
            }
        \end{enumerate}
    \end{definition}

    Each fiber $S_x M$ at $x \in M$ of the spinor bundle $SM$ is equipped with an indefinite inner product $\spinbraket{\cdot}{\cdot}_x$.
    By applying a regularization operator $\mathcal{R}_\varepsilon$, we can extend this inner product to not necessarily
    smooth vectors $u, v \in \H$, i.e.
    \begin{align}
        \label{eq:spin-inner-product}
        (u, v) \mapsto \spinbraket{ (\mathcal{R}_\varepsilon u)(x)}{(\mathcal{R}_\varepsilon v)(x)}_x\,.
    \end{align}
    This product encodes the local correlations of the two fermionic states at point $x$.
    Identifying the fibers with the corresponding spinor spaces of the causal fermion system gives the so-called local
    correlation operators $F^\varepsilon(x)$ defined by
    \begin{align}
        \bra{u} F^\varepsilon(x) \ket{v} = - \spinbraket{ (\mathcal{R}_\varepsilon u)(x)}{(\mathcal{R}_\varepsilon v)(x)}_x \,.
    \end{align}
    Again, the existence and uniqueness of $F^\varepsilon(x)$ as a bounded linear operator on $\H$ is ensured by the Riesz representation theorem.
    Since the spin inner product~\eqref{eq:spin-inner-product} has signature $(n, n)$ with $n = 2^{\lfloor \tfrac{d}{2} \rfloor - 1}$,
    the operator $F^\varepsilon(x)$ is self-adjoint with $n$ positive and $n$ negative eigenvalues~\cite[Theorem 2.15]{neumann}.
    Next, we fix the spin dimension of the causal fermion system to $n$, then $F^\varepsilon$ maps from $M$ to $\F$.
    Consequently, $F^\varepsilon$ induces a regular Borel measure $\rho$ on $\F$ as the push-forward of the volume measure $\mu_M$.
    By construction, the support of $\rho$ is given as the image of $\overline{F^\varepsilon(M)}$.

    Furthermore, $F^\varepsilon$ defines a direct identification between the spacetime $M$ and the support of the measure $\rho$.
    Each point in $M$ corresponds to a linear operator in $\F$ and each operator in the support of $\rho$ corresponds
    to a point in $M$.
    In the following, we do not make this distinction explicitly, but rather implicitly identify points in the support of
    $\rho$ with points in $M$.
    Instead of $P(F^\varepsilon(x), F^\varepsilon(y))$ we write $P^\varepsilon(x, y)$.

    So far, we constructed the triple $(\H, \F, \rho)$ defining a causal fermion system which models the globally hyperbolic
    spacetime $(M, g)$.
    To show that it is a critical point of the causal action, we need to determine the kernel of the fermionic
    projector to compute the closed chain and its eigenvalues.

    We start by applying relation~\eqref{eq:x-psi-relation} to $F^\varepsilon(x)$, which gives the wave evaluation operator
    \begin{align}
        \label{eq:evaluation-operator-for-globally-hyperbolic-spacetime}
        \Psi^\varepsilon(x): \H \to S_x M \,, \quad u \mapsto (\mathcal{R}_\varepsilon u)(x) \,.
    \end{align}
    Its adjoint is given by the following proposition.

    \begin{proposition}
        Let $x \in M$ and $\Psi^\varepsilon(x)$ be the regularized wave evaluation operator as defined above.
        Then its adjoint is given by
        \begin{align}
            \Psi^\varepsilon(x)^*: S_x M \to \H \,, \quad \phi \mapsto \frac{1}{2 \pi} \mathcal{R}_\varepsilon^* A k_m (\cdot, x) \phi
        \end{align}
    \end{proposition}
    \begin{proof}
        Let $u \in \H$ and $f \in \Gamma_c(M, SM)$, then
        \begin{align*}
            \int_M \spinbraket{f(x)}{\Psi^\varepsilon(x) u}_x \dd \mu(x) = \int_M \spinbraket{f(x)}{(\mathcal{R}_\varepsilon u)(x)}_x \dd \mu(x) \,.
        \end{align*}
        By Definition~\ref{def:regularization-operator}, $\mathcal{R}_\varepsilon u$ is a solution of the Dirac equation.
        Thus, we can apply~\cite[Proposition 2.2]{infinite}
        \begin{align*}
            \int_M \spinbraket{f(x)}{(\mathcal{R}_\varepsilon u)(x)}_x \dd \mu(x) &= \frac{1}{2 \pi} (k_m f | \mathcal{R}_\varepsilon u) \\
            &= \frac{1}{2 \pi} \braket{A^{\frac{1}{2}} k_m f}{A^{\frac{1}{2}} \mathcal{R}_\varepsilon u} \\
            &= \frac{1}{2 \pi} \braket{\mathcal{R}_\varepsilon^* A k_m f}{u} \\
            &= \frac{1}{2 \pi} \int_M \braket{\mathcal{R}_\varepsilon^* A k_m(\cdot, x) f(x)}{u} \dd \mu(x) \,.
        \end{align*}
        Since the above expression holds for all $f \in \Gamma_c(M, SM)$, we conclude
        \begin{align*}
            \spinbraket{\phi}{\Psi^\varepsilon(x) u}_x = \frac{1}{2 \pi} \braket{\mathcal{R}_\varepsilon^* A k_m (\cdot, x) \phi}{u}
        \end{align*}
        for all $x \in M$ and $\phi \in S_x M$.
    \end{proof}

    Thus, we get the kernel of the fermionic projector by inserting $\Psi^\varepsilon(x)$ and $\Psi^\varepsilon(x)^*$ into
    Definition~\eqref{eq:kernel-of-the-fermionic-projector}, which gives
    \begin{align}
        P^\varepsilon(x, y) = -\Psi^\varepsilon(x) \Psi^\varepsilon(y)^* = \phi \mapsto -\frac{1}{2 \pi} \left( \mathcal{R}_\varepsilon (\mathcal{R}_\varepsilon^* A k(\cdot, y) \phi) \right)(x) \,.
    \end{align}
    By construction and the definition of the regularization operator, the kernel of the fermionic projector is of
    Hadamard form~\eqref{eq:hadamard-form} for $\varepsilon \searrow 0$ if the quasi-free state $\omega$ fulfills the Hadamard condition.

    \begin{proposition}
        Let $\omega$ be a quasi-free Hadamard state.
        Then the kernel of the fermionic projector $P^\varepsilon(x, y)$ as constructed above is of Hadamard form.
    \end{proposition}
    \begin{proof}
        By definition $\omega$ is a Hadamard state, if for $f, g \in \Gamma_c(M, SM)$ the two point function $\omega(\Psi(f) \Psi(g)^*)$
        can be expressed as an integral of the form
        \begin{align*}
            \omega(\Psi(f) \Psi(g)^*) = \int_M \dd \mu(x) \int_M \dd \mu(y) ~ \spinbraket{f(x)}{H(x, y) g(y)}_x \,,
        \end{align*}
        where $H(x, y)$ is locally of Hadamard form~\eqref{eq:hadamard-form}.
        By definition of the one-particle density operator, we have
        \begin{align*}
            \omega(\Psi(f) \Psi(g)^*) = (k_m f | A | k_m g) = \braket{A (k_m f)}{A (k_m g)} \,.
        \end{align*}
        On the other hand, we have
        \begin{align*}
            \int_M \dd \mu(x) &\int_M \dd \mu(y) \spinbraket{f(x)}{P^\varepsilon(x, y) g(y)}_x \\
            &= - \int_M \dd \mu(x) \int_M \dd \mu(y) \spinbraket{f(x)}{\Psi^\varepsilon(x) \Psi^\varepsilon(y)^* g(y)}_x \\
            &= - \int_M \dd \mu(x) \int_M \dd \mu(y) \braket{\Psi(x)^* f(x)}{\Psi(y)^* g(y)} \\
            &= - \frac{1}{(2 \pi)^2}\braket{\mathcal{R}_\varepsilon A k_m f}{\mathcal{R}_\varepsilon A k_m g} \xrightarrow{\varepsilon \searrow 0} \frac{1}{(2 \pi)^2} \braket{A k_m f}{A k_m g} \,.
        \end{align*}
        Thus, in the limit $\varepsilon \searrow 0$, the kernel of the fermionic projector $P^\varepsilon$ goes to $\tfrac{1}{(2 \pi)^2} H$ in the
        distributional sense, which is of Hadamard form.
    \end{proof}

    Additionally, for $\varepsilon > 0$, $P^\varepsilon$ is a smooth bi-solution of the Dirac equation.
    Therefore, we assume, that for $\varepsilon > 0$, $P^\varepsilon(x, y)$ can be expressed in terms of the $i \varepsilon$ regularized
    Schwinger-DeWitt expansion introduced in Section~\ref{sec:iepsilon-regularization-of-the-schwinger-dewitt-expansion}.
    The regularizing scalar field depends on the choice of the regularization operator $\mathcal{R}_\varepsilon$.
    However, as we will see, the leading geometric singularity structure required for the macroscopic field equations remains universal.

    In the following, we analyze the causal action by considering the fermionic projector to different degrees.
    As stated above, the light cone expansion symbols $T^{(n)}(x, y)$ become singular on the light cone for $\varepsilon \searrow 0$.
    The highest degree then determines the largest contribution to the causal action.
    Starting with the kernel of the fermionic projector, the leading order is given by
    \begin{align}
        \label{eq:fermionic-projector-expansion}
        P^\varepsilon(x, y) &= (i \gamma_x^\mu \nabla^{(1)}_\mu + m) \, G(x, y) \nonumber \\
        &= \frac{i}{2} \slashed{\xi} \Delta^{\frac{1}{2}}(x, y) U(x, y) T^{(-1)}(x, y) + (\deg \leq \tfrac{d}{2} - 1 ) \,,
    \end{align}
    where $(\deg \leq \tfrac{d}{2} - 1)$ collects all contributions of degree less than or equal to $\tfrac{d}{2} - 1$.
    This also includes orders of $\varepsilon$ which come from the regularized van Vleck determinant and spin parallel
    transport map.

    \begin{example}
        From Example~\ref{eg:minkowski-hilbert-space}, we already know that $\H$ is given by the negative energy solution
        of the Dirac equation.
        A possible family of regularization operators is the \textbf{$i\varepsilon$-regularization}, which in momentum space
        is given by a multiplication by the function $p \mapsto e^{\frac{\varepsilon}{2} p_0}$.
        The causal fundamental solution in momentum space is given by
        \begin{align*}
            \hat{k}_m(p) = (\slashed{p} + m) \delta(p^2 - m^2) \epsilon(p_0) \Longrightarrow \widehat{A k_m}(p) = - (\slashed{p} + m) \delta(p^2 - m^2) \Theta(-p_0) \,.
        \end{align*}
        Hence, we have for $u \in \H$ and $\phi \in S_x M$,
        \begin{align*}
            \Psi^\varepsilon(x) u &= \int \frac{\dd^4 p}{(2 \pi)^2} \hat{u}(p) e^{\frac{\varepsilon}{2} p_0} e^{-ipx} \\
            \Psi^\varepsilon(x)^* \phi &= y \mapsto -\int \frac{\dd^4 p}{(2 \pi)^2} (\slashed{p} + m) \delta(p^2 - m^2) \Theta(-p_0) e^{\frac{\varepsilon}{2} p_0} e^{ip(y - x)} \phi \,.
        \end{align*}
        Thus, the kernel of the fermionic projector is given by
        \begin{align*}
            P^\varepsilon(x, y) = \int \frac{\dd^4 p}{(2 \pi)^4} (\slashed{p} + m) \delta(p^2 - m^2) \Theta(-p_0) e^{\varepsilon p_0} e^{ip(y - x)}
        \end{align*}
        A direct computation of the integral gives
        \begin{align*}
            P^\varepsilon(x, y) = \frac{i}{2} \slashed{\xi} T^{(-1)}(x, y) + m T^{(0)}(x, y)
        \end{align*}
        with the regularizing scalar field given by Example~\ref{eg:iepsilon-regularization-minkowski}
    \end{example}

    The expression of the kernel of the fermionic projector can be further simplified, using the covariant expansion of
    the van Vleck-Morette determinant~\cite[Section 2.9]{Visser_1993}
    \begin{align}
        \label{eq:van-vleck-expansion}
        \Delta^{\frac{1}{2}}(x, y) = 1 + \frac{1}{12} R_{\mu \nu}(x) \sigma^\mu \sigma^\nu + \mathcal{O}\left(\sigma^\frac{3}{2}\right) \,.
    \end{align}
    As stated before, each additional factor of order $\sigma$ reduces the degree by one.
    Thus, only the first term of the expansion contributes to the leading order of the fermionic projector, while all the
    other terms are already included in $(\deg \leq \tfrac{d}{2} - 1)$.
    Hence, we have
    \begin{align}
        \label{eq:fermionic-projector-curved-deg-3}
        P^\varepsilon(x, y) &= \frac{i}{2} \slashed{\xi} U(x, y) T^{(-1)}(x, y) + (\deg \leq \tfrac{d}{2} - 1) \\
        &\eqdef P^{(0)}(x, y) +  (\deg \leq \tfrac{d}{2} - 1)\,.
    \end{align}
    For the rest of this section, the causal action is analyzed using $P^{(0)}(x, y)$ while in Section~\ref{sec:perturbative-analysis-of-the-euler-lagrange-equations},
    we then consider the next-to-leading order contribution contained in $(\deg \leq \tfrac{d}{2} - 1)$.

    \subsection{The Closed Chain and its Eigenvalues}\label{subsec:closed-chain-and-its-eigenvalues}

    To check if the causal fermion system constructed in the previous section is a critical point of the causal action,
    we need to compute the eigenvalues of the operator product $xy$ and determine $Q$ (see~\eqref{eq:restricted-euler-largrange-equations}).
    Using the kernel of the fermionic projector up to degree $ < \tfrac{d}{2} - 1$~\eqref{eq:fermionic-projector-curved-deg-3},
    the closed chain can be computed using~\eqref{eq:closed-chain-def} which gives
    \begin{align}
        A_{xy} = \frac{\abs{T^{(-1)}(x, y)}^2}{4} \left( \xi^\mu \complexconj{\xi_\mu} \id_{S_x M} + c_{\mu \nu} \Sigma_x^{\mu \nu} \right) + (\deg < d - 1) \,,
    \end{align}
    where $\Sigma_x^{\mu \nu} = \frac{i}{2} [\gamma_x^\mu, \gamma_x^\nu]$ are the usual bilinear Dirac matrices and
    \begin{align}
        c_{\mu \nu} = \frac{1}{2i} \left( \xi_\mu \complexconj{\xi_\nu} - \xi_\nu \complexconj{\xi_\mu} \right)  \,.
    \end{align}
    Here, we used the fact that the spin parallel transport map satisfies
    \begin{align}
        U(x, y) U(y, x) = \id_{S_x M} \,.
    \end{align}
    In addition, we used the fact that the regularized light cone expansion symbols satisfy
    \begin{align}
        T^{(n)}(y, x) = \complexconj{T^{(n)}(x, y)} \,.
    \end{align}
    By definition of $\xi$, we have
    \begin{align}
        \label{eq:closed-chain-biliean-coefficient}
        \xi &= -\nabla^{(1)} \sigma + i \varepsilon \nabla^{(1)} f \\
        c_{\mu \nu} &= \varepsilon \left( \sigma_\nu \chi_\mu - \sigma_\mu \chi_\nu \right) + \mathcal{O}(\varepsilon^2) \,,
    \end{align}
    with $\chi = - \nabla^{(1)} \Re f$.

    At this point, we want to highlight that the structure of the kernel of the fermionic projector and hence the closed
    chain is very similar to the one in Minkowski spacetime.
    More precisely, $P(x, y) \propto \slashed{\xi}$ and hence the closed chain only contains a scalar and a bi-linear
    component.
    Therefore, we can apply~\cite[Proposition 4.3]{cfs-current} to determine the eigenvalues of the closed chain.
    These are given by
    \begin{align}
        \label{eq:closed-chain-eigenvalues-curved}
        \lambda_{\pm}^{xy} = \frac{\abs{T^{(-1)}(x, y)}^2}{4} \left( \xi^\mu \complexconj{\xi_\mu} \pm \sqrt{ 2 c_{\mu \nu} c^{\mu \nu} } \right) + (\deg < d - 1) \,.
    \end{align}
    As it becomes clear from the following proposition, the eigenvalues of the closed chain form a complex conjugate pair.
    Hence, they all have the same absolute value, and therefore the Lagrangian is of degree less than $d - 1$ for all
    $x$ and $y$ in the convex normal neighborhood $U$.

    \begin{proposition}
        Let $f: U \times U \to \C$ be a regularizing scalar field and $\chi = -\nabla^{(1)} \Re f$, then for all $x, y \in U$, we have
        \begin{align}
            \label{eq:expression-radial-function-squared}
            c_{\mu \nu} c^{\mu \nu} \leq 0 \quad \text{ up to } \quad \mathcal{O}(\varepsilon^3)\,.
        \end{align}
    \end{proposition}
    \begin{proof}
        For $x, y \in U$, we have
        \begin{align*}
            c_{\mu \nu} c^{\mu \nu} &= 2 \varepsilon^2 \left( \sigma_\mu \sigma^\mu \chi_\mu \chi^\mu - (\sigma^\mu \chi_\mu)^2\right)
            + \mathcal{O}(\varepsilon^3) \\
            &= -2 \varepsilon^2 \left( (\sigma^\mu \chi_\mu)^2 - 2 \sigma \chi_\mu \chi^\mu \right) + \mathcal{O}(\varepsilon^3)\,.
        \end{align*}
        If $\sigma = 0$, then
        \begin{align*}
            (\sigma^\mu \chi_\mu)
            ^2 - 2 \sigma \chi_\mu \chi^\mu = (\sigma^\mu \chi_\mu)^2
        \end{align*}
        is the square of a real number and hence non-negative.

        \noindent
        Otherwise, let $u^\mu \defeq \frac{\sigma^\mu}{\sqrt{2 \abs{\sigma}}}$, then $g(u, u) = \sgn \sigma$.
        If $\sigma > 0$, then $g(u, u) = 1$ (i.e.~$u$ is a timelike vector) and
        \begin{align*}
            \left( \sigma^\mu \chi_{\mu} \right)^2 - 2 \sigma \chi_\mu \chi^\mu
            &= 2 \sigma \left( (u^\mu \chi_{\mu})^2 - \chi_\mu \chi^\mu \right) \\
            &= 2 \sigma \left( u \otimes u - g \right)(\chi, \chi)
        \end{align*}
        Since $u$ is timelike, the tensor $u \otimes u - g$ defines a positive semi-definite bilinear form on $T_x M$
        Thus, the above expression is non-negative.

        \noindent
        If $\sigma < 0$, then
        \begin{align*}
            \left(\sigma^{\mu} \chi_\mu(x, y) \right)^2 - 2 \sigma g(\chi, \chi)
            &= 2 \abs{\sigma} \left( (u^\mu \chi_\mu)^2 + \chi_\mu \chi^\mu \right) > 0 \,,
        \end{align*}
        because $\chi$ is by Definition~\ref{def:regularizing-scalar-field} a timelike vector field.
    \end{proof}

    To align the notation with the one used in Minkowski spacetime, we introduce a radial function w.r.t.~regularizing
    scalar field.
    This exactly captures the expression $(\sigma^\mu \chi_\mu)^2 - 2 \sigma \chi_\mu \chi^\mu$ which appears in $c_{\mu \nu} c^{\mu \nu}$
    and is by the previous proposition non-negative.

    \begin{definition}
        \label{def:radial-scalar-field}
        Let $f: U \times U \to \C$ the regularizing scalar field, then the \textbf{temporal scalar field} associated to $f$ is the real
        number
        \begin{align}
            t(x, y) \defeq \sigma^\mu(x, y) \chi_\mu \,,
        \end{align}
        where $\chi \defeq -\nabla^{(1)} \Re f(x, y)$.
        Further, the \textbf{radial scalar field} associated with $f$ is the non-negative real number given by
        \begin{align}
            r(x, y) \defeq \sqrt{t^2(x, y) - 2 \sigma(x, y) \chi_\mu \chi^\mu} \,.
        \end{align}
    \end{definition}

    The physical interpretation of the temporal and radial scalar fields becomes clear when rearranging the expression as
    \begin{align}
        2 \sigma(x, y) = \frac{1}{g_x(\chi, \chi)} \left( t^2(x, y) - r^2(x, y) \right) \,.
    \end{align}
    Moreover, using the radial function in~\eqref{eq:closed-chain-eigenvalues-curved} for the eigenvalues of the
    closed chain gives
    \begin{align}
        \label{eq:eigenvalues-in-continuum}
        \lambda^{xy}_{\pm} = \frac{\abs{T^{(-1)}(x, y)}^2}{4} \left( \xi^\mu \complexconj{\xi_\mu} \pm 2 i \varepsilon r(x, y) \right) + (\deg < d - 1) \,.
    \end{align}
    For $r(x, y) \neq 0$, $A_{xy}$ is diagonalizable~\cite[Proposition 4.4]{cfs-current} with eigenspace projectors
    \begin{align}
        \label{eq:degenerate-eigenspace-projectors}
        \Lambda^{xy}_{\pm} = \frac{\id_{S_x M}}{2} \pm \frac{c_{\mu \nu} \Sigma^{\mu \nu}}{4 i \varepsilon r(x, y)}\,.
    \end{align}
    Note that the eigenvalues have a $n$-fold degeneracy, hence the corresponding projectors have rank $n$.

    To completely decompose $\Lambda^{xy}_{\pm}$ into rank-$1$ projectors, one requires a set of $\log_2 n = \lfloor d/2 \rfloor - 1$
    mutually commuting projectors.
    This decomposition can be made explicit by analyzing the bivector term $c_{\mu\nu} \Sigma^{\mu\nu}$, which acts exclusively
    in the two-dimensional plane spanned by $\sigma^\mu$ and $\chi^\mu$ (see~\eqref{eq:closed-chain-biliean-coefficient}).
    Consequently, for any vectors $u$ and $v$ chosen orthogonal to both $\sigma^\mu$ and $\chi^\mu$, the corresponding bivector operator
    $u_\mu v_\nu \Sigma^{\mu \nu}$ commutes with $\Lambda^{xy}_{\pm}$.
    By selecting $\lfloor d/2 \rfloor - 1$ mutually orthogonal pairs $(u^{(k)}, v^{(k)})$ from the orthogonal complement
    of $\Span\{\sigma^\mu, \chi^\mu\}$, one obtains the requisite set of mutually commuting operators.
    These can then be used to construct the commuting projectors necessary for a full rank-$1$ decomposition.
    In general the vectors $u^{(k)}, v^{(k)}$ depend on $x$ and $y$.
    In $d = 4$ dimensions ($n = 2$), chiral projectors constructed from $\gamma^5_x$ naturally commute with $\Lambda^{xy}_{\pm}$ and
    only depend on $x$.
    Thus, they define a decomposition of the eigenspaces independently of $y$.

    These results are valid only within a convex normal neighborhood $U$.
    Extending them globally would require a detailed understanding of the fermionic projector's singular structure,
    particularly at light-cone caustics where the light-cone expansion breaks down.
    Nevertheless, because $P(x, y)$ and the closed chain $A_{xy}$ exhibit strong singularities at $x = y$, we assume that
    the dominant contributions to the \ac{el} equations arise from a small neighborhood around the diagonal.
    Consequently, we approximate the integration over the entire manifold $M$ by restricting it to $U$.

    We conclude, that the kernel of the operator $Q$ in the restricted \ac{el} equations~\eqref{eq:restricted-euler-largrange-equations}
    is of degree less than $\tfrac{3}{2}(d - 1)$, i.e.
    \begin{align}
        Q(x, y) = \left(\deg < \tfrac{3}{2}(d - 1) \right) \,.
    \end{align}
    Thus, every globally hyperbolic spacetime is a critical point of the causal action at degree $\tfrac{3}{2}(d - 1)$.

    At the highest singular degree, the causal action does not constrain the geometry of the spacetime $(M, g)$.
    Consequently, an arbitrary globally hyperbolic spacetime is not necessarily a solution to the vacuum Einstein equations,
    meaning that from the perspective of general relativity, it may not represent a valid physical vacuum.
    However, as we will demonstrate in Section~\ref{subsec:perturbations-with-a-fermionic-particle}, gravity emerges strictly
    as a higher-order effect in $\varepsilon$.
    The leading singularity isolates purely non-gravitational vacuum fluctuations and therefore cannot encode constraints
    on the macroscopic curvature of $M$.
    At this order, the effective framework is simply a \ac{qft} on a fixed globally hyperbolic background without dynamical
    gravity.
    This fact can be seen by the structural similarity of the fermionic projector~\eqref{eq:fermionic-projector-curved-deg-3}
    to its Minkowski counterpart.
    To recover gravitational dynamics, we must therefore resolve the next-to-leading-order contributions, which we analyze
    in the following section.

    \section{Perturbative Analysis of the Euler-Lagrange Equations}\label{sec:perturbative-analysis-of-the-euler-lagrange-equations}

    For the analysis of the causal action principle to degree $d - 1$, we consider perturbations of the fermionic projector
    \begin{align}
        P(x, y) = P^{(0)}(x, y) + \delta P(x, y)\,.
    \end{align}
    This perturbation corresponds to a variation of the underlying wave functions $\Psi$.
    Then, the linearized field equations~\eqref{eq:linearized-field-equations} constrain the perturbation $\delta P$.
    In the following, we consider different types of perturbations and analyze the resulting constraints.
    We start with the next-to-leading order term in the Schwinger-DeWitt expansion of the fermionic projector,
    which includes curvature contributions.
    Next, we also consider the perturbation by a fermionic particle, which then requires a coupling of matter to curvature.

    \subsection{Vacuum Euler Lagrange Equations to Degree $d - 1$}\label{subsec:geometric-perturbations-at-order-epsilon-squared}

    In Section~\ref{subsec:closed-chain-and-its-eigenvalues}, we showed that in the continuum limit, every globally hyperbolic
    spacetime gives a critical point of the causal action.
    In the following, we analyze how next-to-leading singularities contribute to the restricted \ac{el} equations.
    For this, we compute $P(x, y)$ up to degree $\tfrac{d}{2} - 1$.
    These terms also include derivatives of the van Vleck-Morette determinant $\Delta(x, y)$ (see Expansion~\eqref{eq:van-vleck-expansion})
    \begin{align}
        \nabla^{(x)}_{\mu} \Delta^{\frac{1}{2}}(x, y) &= \frac{1}{6} R_{\mu \nu}(x) \sigma^{\nu} + \mathcal{O}(\sigma) \,, \\
        \Box^{(1)} \Delta^{\frac{1}{2}}(x, y) &= \frac{1}{6} R(x) + \mathcal{O}(\sigma^\frac{1}{2}) \,.
    \end{align}
    For the derivative of the parallel transport map, we use the following lemma.
    \begin{lemma}
        Let $\nabla$ be a free spin connection and $U(x, y): S_y M \to S_x M$ be a spin parallel transport map.
        Then
        \begin{align}
            \nabla^{(1)}_\mu U(x, y) = \left( -\frac{i}{8} R_{\alpha \beta \mu \nu}(x) \sigma^\nu \Sigma^{\alpha \beta}_x + \mathcal{O}(\sigma) \right) U(x, y) \,,
        \end{align}
        where $R_{\alpha \beta \mu \nu}(x)$ denotes the Riemann curvature tensor at $x$.
    \end{lemma}
    \begin{proof}
        The spin parallel transport map is defined by
        \begin{align*}
            U(x, x) = \id_{S_x M} \quad \text{and} \quad \sigma^\mu \nabla^{(1)}_\mu U(x, y) = 0
        \end{align*}
        Taking the derivative of the second equation with respect to $x$ gives
        \begin{align*}
            0 = \nabla^{(1)}_\nu \sigma^\mu \nabla^{(1)}_\mu U(x, y) + \sigma^\mu \nabla^{(1)}_\nu \nabla^{(1)}_\mu U(x, y) \,.
        \end{align*}
        Evaluating this expression at $x = y$, we get $\nabla^{(1)}_\nu U(x, x) = 0$.
        Applying a second derivative w.r.t to $x$ and using the same argument, we also get
        \begin{align*}
            0 &= \left(\nabla^{(1)}_\mu \nabla^{(1)}_\nu + \nabla^{(1)}_\nu \nabla^{(1)}_\mu\right) U(x, x) \\
            &\Rightarrow \nabla^{(1)}_\mu \nabla^{(1)}_\nu U(x, x) = \frac{1}{2} \left[\nabla^{(1)}_\mu, \nabla^{(1)}_\nu\right] U(x, x) = \frac{i}{8} R_{\alpha \beta \mu \nu}(x) \Sigma^{\alpha \beta}_x \,,
        \end{align*}
        where for the last equality we used that $\nabla$ is a torsion-free free spin connection and~\cite[Theorem 4.2.5]{intro}.
        Using the above results, we can compute the leading terms of the Taylor expansion of $\nabla^{(1)}_\mu U(x, y) U(y, x)$ in $y$, which
        gives the proposed result.
    \end{proof}

    Thus, the second coefficient $a_1(x, y)$ of the Schwinger-DeWitt expansion~\eqref{eq:schwinger-dewitt-expansion}
    \begin{align}
        a_1(x, y) &= U(x, y) \int_0^1 U(y, \gamma(s)) \Delta^{-\frac{1}{2}}(\gamma(s), y) \Box^{(1)} \left( \Delta^{\frac{1}{2}}(\gamma(s), y) U(\gamma(s), y) \right) ~ \dd s \nonumber \\
        &\quad- \frac{1}{4} U(x, y) \int_0^1 R(\gamma(s)) ~ \dd s \nonumber \\
        &= -\frac{1}{12} U(x, y) \int_0^1 R(\gamma(s)) ~ \dd s + \mathcal{O}(\sigma^\frac{1}{2})
    \end{align}
    where $\gamma$ is the unique geodesic connecting $x = \gamma(0)$ to $y = \gamma(1)$.
    The integral can also be expanded around $x$, which then gives
    \begin{align}
        a_1(x, y) = -\frac{1}{12} R(x) U(x, y) + \mathcal{O}(\sigma^\frac{1}{2}) \,.
    \end{align}
    Using the Schwinger-DeWitt expansion~\eqref{eq:schwinger-dewitt-expansion} for $G(x, y)$ up to degree $\tfrac{d}{2} - 1$,
    we compute the kernel of the fermionic projector again by applying $i \gamma^\mu_x \nabla^{(1)}_\mu + m$.
    This gives
    \begin{align}
        \label{eq:fermionic-projector-curved-deg-2}
        P(x, y) &= \left(1 + \frac{1}{12} R_{\mu\nu}(x) \xi^\mu \xi^\nu \right) P^{(0)}(x, y) \nonumber \\
        &\quad+ m U(x, y) T^{(0)}(x, y) \nonumber \\
        &\quad- \frac{i}{24} R(x) \slashed{\xi} U(x, y) T^{(0)} (x, y) - \frac{5i}{12} \gamma^{\mu}_x R_{\mu \nu}(x) \xi^{\nu} U(x, y) T^{(0)} (x, y) \nonumber \\
        &\quad+ \left(\deg \leq \tfrac{d}{2} - 2\right) \,.
    \end{align}
    Note that terms which are proportional to $\slashed{\xi}$ and thus to $P^{(0)}(x, y)$ only contribute by an overall
    scaling of the eigenvalues of the closed chain.
    Thus, these terms do not contribute to the causal action and hence to the linearized field equations.
    In addition, the mass term give rise to a vectorial component in the closed chain.
    However, a vectorial perturbation of the closed chain does not contribute due to $\delta \lambda_i \propto \tr\left[ \Lambda_i \gamma^\mu \right] = 0$.
    Thus, the only non-vanishing contribution to the linearized field equations is given by
    \begin{align} \label{eq:geometric-perturbation}
        \delta P^{\text{geom}}(x, y) = - \frac{5i}{12} \left( R_{\mu\nu} - \frac{1}{d} R(x) g_{\mu \nu}(x) \right) \gamma^{\mu}_x \xi^{\nu} U(x, y) T^{(0)} (x, y) \,.
    \end{align}
    We conclude that the causal fermion system up to degree $\tfrac{3}{2} d - 2$ satisfies the restricted \ac{el} equations:

    \begin{theorem} \label{thm:trace-free-vacuum-einstein-equations}
        Let $(M, g)$ be a globally hyperbolic spacetime of dimension $d > 2$ and $(\H, \F, \rho)$ the corresponding causal
        fermion system.
        Then, the geometric perturbation $\delta P^{\text{geom}}(x, y)$ defined by~\eqref{eq:geometric-perturbation}
        satisfies the linearized field equations~\eqref{eq:linearized-field-equations} up to degree $\tfrac{3}{2} d - 2$ if
        \begin{align}
            R^{(TF)}_{\mu \nu} = 0 \,,
        \end{align}
        where $R^{(TF)}_{\mu \nu} = R_{\mu\nu} - \tfrac{1}{d} R(x) g_{\mu \nu}(x)$.
    \end{theorem}
    \begin{proof}
        Asymptotically for small $\varepsilon$, the linearized field equations are given by $(\delta Q \Psi)(x) = 0$.
        Multiplying the linearized field equation $(\delta Q \Psi)(x) = 0$ by $A^\mu \Psi^*(x)$ gives $(\delta Q \Psi)(x) A^\mu \Psi^*(x) = 0$.
        An explicit choice for the operator $A^\mu$ and the full calculation of this expression are given in Appendix~\ref{sec:analysis-of-contribution-to-the-linearized-field-equations}.
        It implies $R^{(TF)}_{\mu \nu}(x) = 0$.
    \end{proof}

    A standard result is that once the trace-free vacuum Einstein equations are satisfied, then the full
    Einstein equations are satisfied up to an integration constant $\Lambda$.

    \begin{corollary}\label{cor:vacuum-einstein-equation}
        Let $(M, g)$ be a globally hyperbolic spacetime of dimension $d > 2$ satisfying $R^{(TF)}_{\mu \nu} = 0$, then
        \begin{align}
            \label{eq:vacuum-einstein}
            R_{\mu \nu} - \frac{1}{2} R g_{\mu \nu} + \Lambda g_{\mu \nu} = 0
        \end{align}
        for some constant $\Lambda$.
    \end{corollary}
    \begin{proof}
        By definition, we have
        \begin{align*}
            0 = R^{(TF)}_{\mu \nu} = R_{\mu \nu} - \frac{1}{d} R g_{\mu \nu} \,,
        \end{align*}
        By the contracted Bianchi identity, we have
        \begin{align*}
            0 = \nabla^\nu \left(R_{\nu \mu} - \frac{1}{2} g_{\nu \mu} R \right) = \left(\frac{1}{d} - \frac{1}{2}\right) \nabla_\nu R \,.
        \end{align*}
        Thus, $R$ must be constant.
        Defining $\Lambda \defeq \frac{d - 2}{2d} R$ then concludes the proof.
    \end{proof}

    In the next section, we introduce matter perturbations to the fermionic projector and analyze how those
    contribute to the restricted \ac{el} equations.

    \subsection{Matter Perturbations by Fermionic Particles}\label{subsec:perturbations-with-a-fermionic-particle}

    As a matter perturbation, we consider the inclusion of a single wave function $u \in \H^\perp$, where the orthogonal
    complement is taken w.r.t.~the scalar product of the full Hilbert space $\H_m$ (see Section~\ref{subsec:the-dirac-equation-in-curved-spacetimes-2}).
    Physically, since $\H$ represents the occupied fermionic background modes (analogous to the Dirac sea), selecting $u$
    from the orthogonal complement corresponds to introducing a real particle into a previously unoccupied state.
    Because $u \in \H_m$, it is an exact solution of the Dirac equation, mirroring a standard positive-energy particle
    excitation in flat spacetime.
    The corresponding perturbation of the fermionic projector is then defined as
    \begin{align}
        \delta P(x, y) &= \frac{1}{2 \pi} \spinket{u(x)} \spinbra{u(y)} \nonumber \\
        &\approx \frac{1}{2 \pi} \spinket{u(x)} \spinbra{(1 + \sigma^\mu \nabla_\mu + \mathcal{O}(\sigma)) u(x)} U(x, y) \,.
    \end{align}
    While in general, all components of $\delta P(x, y)$ contribute to the linearized field equations, we only consider
    the components which are of the same form as $\delta P^{\text{geom}}(x, y)$, i.e.~proportional to $\gamma^\mu \sigma^\nu$.
    The remaining contributions have to be compensated by other effects, which are not treated in this paper.
    The interested reader is referred to~\cite{cfs}.
    Therefore, we only consider
    \begin{align}
        \label{eq:matter-perturbation}
        \delta P^{\text{matter}}(x, y) &= \frac{i}{8 \pi} T^{(TF)}_{\mu \nu}(x) \gamma^\mu_x \sigma^\nu U(x, y) \,,
    \end{align}
    where
    \begin{align}
        T_{\mu \nu}(x) &= \frac{i}{4} \big( \spinbra{u(x)} \gamma_\mu \spinket{\nabla_\nu u(x)} + \spinbra{u(x)} \gamma_\nu \spinket{\nabla_\mu u(x)} \nonumber \\
        &\quad- \spinbra{\nabla_\nu u(x)} \gamma_\mu \spinket{u(x)} - \spinbra{\nabla_\mu u(x)} \gamma_\nu \spinket{u(x)} \big)
    \end{align}
    is the canonical stress energy-momentum tensor of $u(x)$.

    \begin{theorem}\label{thm:trace-free-einstein-equations-with-matter}
        Let $(M, g)$ be a globally hyperbolic spacetime of dimension $d > 2$ and $(\H, \F, \rho)$ the corresponding causal
        fermion system.
        Further, let $u \in \H^\perp$ be solution of the Dirac Equation~\eqref{eq:dirac-equation}
        and $T_{\mu \nu}$ the corresponding canonical energy-momentum tensor.
        Then, asymptotically for small $\varepsilon$ the perturbation $\delta P^{\text{geom}} + \delta P^{\text{matter}}$
        satisfies the linearized field equations~\eqref{eq:linearized-field-equations} up to degree $\tfrac{3}{2} d - 2$
        if
        \begin{align}
            R^{(TF)}_{\mu \nu}(x) &= 8 \pi G(x) \, T^{(TF)}_{(\mu \nu)}(x) + \mathcal{O}(\varepsilon^3) \,, \label{eq:einstein-equations-with-matter}
        \end{align}
        where $T^{(TF)}_{(\mu \nu)}(x) = T_{\mu \nu}(x) - \frac{1}{d} T(x) g_{\mu \nu}(x)$
        with $T(x) = T^{\mu}_{\mu}(x)$ is the trace-free energy momentum tensor and the coupling constant $G(x)$ is of order $\mathcal{O}(\varepsilon^2)$.
    \end{theorem}
    \begin{proof}
        The proof is analogous to the proof for Theorem~\ref{thm:trace-free-vacuum-einstein-equations}.
    \end{proof}

    The proportionality factor $G \propto \varepsilon^2$, which acts as the effective gravitational coupling constant,
    is determined entirely by the regularization scheme.
    In particular, if the regularizing scalar field $f(x, y)$ is locally translation-invariant, i.e.~if it can be written as
    \begin{align}
        f(x, \exp_x(\xi)) = f(\xi) \quad \text{for all } x \in U,
    \end{align}
    then $G$ evaluates to a strict constant.
    If, however, the background geometry is highly dynamical, such that the regularization intrinsically depends on the
    base point $x$, $G(x)$ would acquire a spacetime dependence.
    This would introduce non-minimal gravitational couplings or a running effective gravitational constant, bearing conceptual
    similarities to scalar-tensor theories of gravity like Brans-Dicke theory~\cite{Brans1961}.

    Theorem~\ref{thm:trace-free-einstein-equations-with-matter} is stated for a single excitation $u \in \H^\perp$.
    Since the linearized field equations~\eqref{eq:linearized-field-equations} are linear in $\delta P$ and the map $u \mapsto T_{\mu\nu}[u]$
    is additive on mutually incoherent (orthogonal) modes, a macroscopic stress tensor is obtained by summing, respectively taking
    the statistical expectation, over many such excitations, $T^{\mathrm{macro}}_{\mu\nu} = \sum_k T_{\mu\nu}[u_k]$.
    This superposition is valid as long as the perturbation stays within the linearized regime and interference (off-diagonal) terms
    between modes are negligible.

    By construction, the symmetrized energy momentum tensor
    \begin{align}
        T_{(\mu \nu)} = \frac{1}{2} \left( T_{\mu \nu} + T_{\nu \mu} \right)
    \end{align}
    is divergence-free.
    Assuming a constant gravitational coupling $G$ for the present derivation, then similar to Corollary~\ref{cor:vacuum-einstein-equation},
    this implies the Einstein equations coupled to matter.

    \begin{corollary}\label{cor:einstein-matter-equation}
        Let $(M, g)$ be a globally hyperbolic spacetime of dimension $d > 2$ satisfying~\eqref{eq:einstein-equations-with-matter} with
        constant gravitational coupling $G$, then
        \begin{align} \label{eq:einstein}
            R_{\mu \nu} - \frac{1}{2} R g_{\mu \nu} + \Lambda g_{\mu \nu} = 8 \pi G ~ T_{(\mu \nu)} \,.
        \end{align}
        for some constant $\Lambda$.
    \end{corollary}
    \begin{proof}
        Using first the Bianchi identity and then~\eqref{eq:einstein-equations-with-matter} gives
        \begin{align*}
            0 &= \nabla^\mu \left(R_{\mu \nu} - \frac{1}{2} R g_{\mu \nu} \right)
            = \nabla^\mu \left(8 \pi G ~ T^{(TF)}_{(\mu \nu)} - \frac{d - 2}{2d} R g_{\mu \nu} \right) \\
            &= \nabla^\mu \left(8 \pi G ~ T_{(\mu \nu)} - \frac{8 \pi G}{d} T g_{\mu \nu} - \frac{d - 2}{2 d} R g_{\mu \nu} \right) \,.
        \end{align*}
        Because $u$ is a solution to the Dirac equation, its symmetrized energy momentum tensor is conserved, $\nabla^\mu T_{(\mu \nu)} = 0$.
        Thus,
        \begin{align*}
            \Lambda \defeq \frac{d - 2}{2 d} R + \frac{8 \pi G}{d} T
        \end{align*}
        is constant.
        In combination with~\eqref{eq:einstein-equations-with-matter}, this implies~\eqref{eq:einstein}.
    \end{proof}

    This completes the perturbative analysis of the restricted \ac{el} equations up to degree $d - 1$.
    We showed that every globally hyperbolic spacetime is a critical point of the causal action in the continuum limit.
    Considering the next-to-leading order contributions, we derived the vacuum Einstein equations.
    Including matter perturbations, we derived the Einstein equations coupled to matter.

    \section{Discussion and Outlook}\label{sec:discussion-and-outlook}

    In this paper, we generalized the framework of causal fermion systems to arbitrary globally hyperbolic
    spacetimes and derived the classical gravitational dynamics from the causal action principle.
    A conceptual challenge in formulating quantum field theory in curved spacetimes is the absence of a preferred
    vacuum state, making it necessary to work with the class of Hadamard states.
    Furthermore, our identification required the introduction of a regularization operator $\mathcal{R}_\varepsilon$.
    The specific choice of this operator may vary depending on initial data, which in turn relies on the observer's frame or foliation.
    However, the causal action principle is independent of the specific choice.
    While different observers might not agree on a common background state or the exact form of the regularizing scalar field,
    the dynamics of the macroscopic perturbations governed by the linearized field equations remain invariant.
    The geometry and the resulting Einstein equations describe the objective, observer-independent physical dynamics.

    In addition, the presented framework provides a systematic method for working out higher-order corrections to the Einstein equations,
    including potential Planck-scale effects.

    The perturbative analysis presented in this paper only gives rise to the Einstein equations
    up to an undetermined cosmological constant $\Lambda$.
    The origin of this constant has profound physical significance.
    Hence, resolving how causal action dynamically generates $\Lambda$ is crucial for connecting the microscopic
    quantum structure to macroscopic cosmological phenomena.

    Another open question for future work is whether the assumption that the unperturbed Dirac equation holds exactly,
    i.e.~that the regularization operators map into the solution space of the Dirac equation, is overly restrictive or physically fully justified.
    This is a strong constraint for the dynamics of the regularization itself.
    However, alternative approaches such as those considered in the context of baryogenesis~\cite{baryogenesis, baryomink, baryoconform}
    suggest that one should treat the regularizing vector field as an independent dynamical entity.
    A major objective for future research will be to determine which approach captures the complete physical picture, and
    to systematically derive the fully coupled higher-order corrections to both the Einstein and Dirac equations.

    \appendix

    \section{Analysis of Contributions to the Linearized Field Equations}\label{sec:analysis-of-contribution-to-the-linearized-field-equations}

    This section contains the computation of the contributions of the geometric and matter perturbations to the linearized
    field equations~\eqref{eq:linearized-field-equations}.
    In particular, we consider the perturbation $\delta P^{\text{geom}}$ defined by~\eqref{eq:geometric-perturbation}
    and $\delta P^{\text{matter}}$ defined by~\eqref{eq:matter-perturbation}.
    For both perturbations, we compute
    \begin{align}
        \delta A_{xy} &= P^{(0)}(x, y) \delta P(y, x) + \delta P(x, y) P^{(0)}(y, x) = A + A_{\mu \nu} \Sigma^{\mu \nu} \,,
    \end{align}
    where $P^{(0)}$ is the fermionic projector up to degree $\frac{d}{2} - 1$ (see~\eqref{eq:fermionic-projector-curved-deg-3}).
    The coefficients $A$ and $A_{\mu \nu}$ are real-valued, and for the two perturbations, they are given by
    \begin{align}
        A^{\text{geom}} &= -\frac{5}{12} R^{(TF)}_{\nu \rho}(x) \Re \left[ \xi^\nu \complexconj{\xi^\rho} T^{(-1)}(x, y) \complexconj{T^{(0)}(x, y)} \right] \,, \\
        A_{\mu \nu}^{\text{geom}} &= -\frac{5}{12} R^{(TF)}_{\nu \rho}(x) \Im \left[ \xi_\mu \complexconj{\xi^\rho} T^{(-1)}(x, y) \complexconj{T^{(0)}(x, y)} \right] \,, \\
        A^{\text{matter}} &= \frac{1}{4 \pi} T^{(TF)}_{\nu \rho}(x) \Re \left[ \xi^\nu \sigma^\rho T^{(-1)}(x, y) \right] \,, \\
        A_{\mu \nu}^{\text{matter}} &= \frac{1}{4 \pi} T^{(TF)}_{\nu \rho}(x) \Im \left[ \xi_\mu \sigma^\rho T^{(-1)}(x, y) \right] \,.
    \end{align}

    From the perturbation of the closed chain $\delta A_{xy}$, we directly compute the corresponding perturbation of the eigenvalues as
    \begin{align}
        \delta \lambda_{\pm} &= \frac{1}{n} \Tr \left[ \Lambda_{\pm} \delta A_{xy} \right] \\
        &= A \pm \frac{c_{\mu \nu}}{\varepsilon r} A^{\mu\nu}
    \end{align}
    with the eigenspace projectors $\Lambda_{\pm}$~\eqref{eq:degenerate-eigenspace-projectors}.
    Thus, the perturbation of the absolute value of the eigenvalues is given by
    \begin{align}
        \delta \abs{\lambda_{\pm}} &= 2 \Re \left[ \frac{\complexconj{\lambda_\pm}}{\abs{\lambda_\pm}} \delta \lambda_{\pm}  \right] \\
        &= \frac{\abs{T^{(-1)}}^2}{2 \abs{\lambda}} \left[ \xi_\alpha \complexconj{\xi^\alpha} \left( A \pm \frac{c_{\mu \nu}}{\varepsilon r} A^{\mu\nu} \right) \right] \,.
    \end{align}
    Evaluating the sum over $j$ in the expression for the perturbation of $\delta Q$ then gives
    \begin{align}
        \sum_j \left( \delta \abs{\lambda_\pm} - \delta \abs{\lambda_j} \right) = \pm \frac{2 \abs{T^{(-1)}}^2}{\abs{\lambda}} \left[ \xi_\alpha \complexconj{\xi^\alpha} \frac{c_{\mu \nu}}{\varepsilon r} A^{\mu\nu} \right] \,.
    \end{align}
    From here, we use the implication
    \begin{align}
        (\delta Q \Psi)(x) = 0 \quad \Longrightarrow \quad 0 &= (\delta Q \Psi)(x) A \Psi^*(x) \text{ for any } A \in \Lin(\H)
    \end{align}

    \begin{lemma} \label{lem:operator-representation-of-perturbation}
        There exists a bounded operator $A^\mu \in \Lin(\H)$ such that
        \begin{align}
            \label{eq:operator-representation-of-perturbation}
            -\Psi(x) A^\mu \Psi^*(y) = P^{(0)}(x, y) \frac{\xi^\mu}{t} + \left(\deg \leq \tfrac{d}{2} - 1 \right) \,.
        \end{align}
    \end{lemma}
    \begin{proof}
        By Definition~\ref{def:regularizing-scalar-field}, $\chi = -\nabla^{(1)} \Re f(x, y)$ is a time-like vector.
        In coordinates with $\chi = (1, 0, \dots, 0)$, we then get from~\eqref{eq:fermionic-projector-curved-deg-3}
        \begin{align*}
            P^{(0)}(x, y) &= \frac{i}{2} \slashed{\xi} U(x, y) T^{(-1)}(x, y) \\
            &= \gamma^\mu_x U(x, y) \int \frac{\dd^{d - 1} k}{2 \omega_k ~ (2 \pi)^d} ~ k_\mu e^{i k_{\nu} \xi^\nu}
        \end{align*}
        with $\omega_k = \sqrt{m^2 + k_j k^j}$ and $k_0 = - \omega_k$.
        In momentum-space, the multiplication by $-\tfrac{k_\mu}{\chi_\nu k^\nu} = \tfrac{k_\mu}{\omega_k}$ defines a
        linear operator $A^\mu$ on $\H$.
        The kernel of the operator $A^\mu$ is given by
        \begin{align*}
            -\Psi(x) A^\mu \Psi^*(y) = \gamma^\nu_x U(x, y) \int \frac{\dd^{d - 1} k}{2 \omega_k ~ (2 \pi)^d} ~ \frac{k_\nu k^\mu}{\omega_k} e^{i k_{\nu} \xi^\nu}  \,.
        \end{align*}
        The operator $A^0$ is the identity operator, while the spatial components $A^i$ are bounded because $\abs{\tfrac{k^j}{\omega_k}} < 1$.
        Since in the chosen frame $t = \chi_\nu \xi^\nu = \xi^0$,~\eqref{eq:operator-representation-of-perturbation} holds
        exactly.
        The spatial remainder is given by
        \begin{align*}
            E^i(x, y) &= \gamma^\mu_x U(x, y) \int \frac{\dd^{d - 1} k}{2 \omega_k ~ (2 \pi)^d} ~ k_\mu \left( \frac{\xi^i}{t} - \frac{k^i}{\omega_k} \right) e^{i k_\nu \xi^\nu} \,.
        \end{align*}
        The leading light-cone singularity of $E^i$ is governed entirely by the high-momentum regime ($\abs{k} \to \infty$)
        where the phase $i k_\nu \xi^\nu$ is stationary.
        The stationary phase condition gives
        \begin{align*}
            0 = \pdv{k^i} \left(i k_\nu \xi^\nu \right) = it \left( \frac{\xi_i}{t} - \frac{k_i}{\omega_k} \right) \,.
        \end{align*}
        Hence, the integrand of $E^i(x, y)$ vanishes linearly on the stationary phase locus.
        The resulting position-space singularity of $E^i(x, y)$ must be strictly lower than the leading singularity of $P^{(0)}(x, y)$.
        Thus,
        \begin{align*}
            \deg E^i < \deg P^{(0)} = \tfrac{d - 1}{2} \,,
        \end{align*}
        which gives the claimed singular structure of the error term in~\eqref{eq:operator-representation-of-perturbation}.
    \end{proof}

    Using the operators $A^\mu$ from Lemma~\ref{lem:operator-representation-of-perturbation} gives
    \begin{align}
        0 &= (\delta Q \Psi)(x) A^\mu \Psi^*(x) \nonumber \\
        &= -\int \dd \lambda(y) ~ \delta Q(x, y) P^{(0)}(y, x) \frac{\complexconj{\xi^\mu}}{t} + \left( \deg \leq \tfrac{3}{2} d - \tfrac{5}{2} \right) \,.
    \end{align}
    Due to the multiplication with $P^{(0)}(y, x)$, the integrand simplifies as
    \begin{align}
        \delta Q(x, y) P^{(0)}(y, x) &= \sum_{i, j} \left( \delta \abs{\lambda_i} - \delta \abs{\lambda_j} \right) \frac{\complexconj{\lambda}}{\abs{\lambda_i}} \Lambda_i P^{(0)}(x, y) P^{(0)}(y, x) \nonumber \\
        &= \sum_{i, j} \left( \delta \abs{\lambda_i} - \delta \abs{\lambda_j} \right) \abs{\lambda} \Lambda_i \nonumber \\
        &= 2 \abs{T^{(-1)}}^2 \sum_{\pm} \left( \pm \xi_\rho \complexconj{\xi^\rho} \frac{c_{\kappa \nu}}{\varepsilon r} A^{\kappa\nu} \right) \Lambda_{\pm} \nonumber \\
        &= - i \abs{T^{(-1)}}^2 \left( \xi_\rho \complexconj{\xi^\rho} \frac{c_{\kappa \nu}}{\varepsilon r} \frac{c_{\alpha \beta}}{\varepsilon r} A^{\kappa\nu} \Sigma^{\alpha \beta} \right) \,.
    \end{align}
    A necessary condition for the linearized field equations to be satisfied up to degree $\deg < 2d - \tfrac{5}{2}$ at $x$ is
    \begin{align}
        0 &= \int \dd \lambda(y) ~ \abs{T^{(-1)}}^2 \xi_\alpha \complexconj{\xi^\alpha} \frac{c_{\sigma  \nu}}{\varepsilon r} \frac{c_{\rho \sigma}}{\varepsilon r} A^{\sigma \nu} \frac{\complexconj{\xi^\mu}}{t}
    \end{align}
    for all $\alpha, \beta$ and $\mu$.
    Instead of integrating over a convex normal neighborhood, we integrate over a corresponding region in the tangent space
    at $x$, since $\Delta^{-1}(x, y) \approx 1$.
    In addition, we use that asymptotically for small $\varepsilon$, we have $\xi^\mu, \complexconj{\xi^\mu} = \sigma^\mu + \mathcal{O}(\varepsilon)$.
    Thus, for the geometric perturbation, we get
    \begin{align}
        0 &= -\frac{5}{12} R^{(TF)}_{\kappa \nu}(x) \int \dd^d \xi ~ \abs{T^{(-1)}}^2 \frac{\sigma}{r^2} \left( \chi_\alpha \xi_\beta - \xi_\alpha \chi_\beta \right) \nonumber \\
        &\qquad \qquad \cdot \left( 2 \sigma \chi^\kappa \xi^\nu - t \xi^\kappa \xi^\nu \right) \Re \left[ T^{(-1)} \complexconj{T^{(0)}} \right] \frac{\xi^\mu}{t} \,.
    \end{align}
    Carrying out the angular integration in the rest frame $\chi = (1, 0, \dots, 0)$, writing
    $\xi = (t, r \hat{n})$ and using
    \begin{align}
        \int_{S^{d - 2}} \hat{n}_i \hat{n}_j ~ \dd\Omega = \frac{\vol(S^{d - 2})}{d - 1} \delta_{ij}
    \end{align}
    and
    \begin{align}
        \int_{S^{d - 2}} \hat{n}_i \hat{n}_j \hat{n}_k \hat{n}_l ~ \dd \Omega = \frac{\vol(S^{d - 2})}{d^2 - 1} \left( \delta_{ij} \delta_{kl} + \delta_{ik} \delta_{jl} + \delta_{il} \delta_{jk} \right)
    \end{align}
    the necessary condition reduces to
    \begin{align}
        0 &= C^\text{geom}(\varepsilon, x) (d + 1) R^{(TF)}_{0j} \,, \\
        0 &= C^\text{geom}(\varepsilon, x) \left( (d + 2) \delta_{ij} R^{(TF)}_{00} + 2 R^{(TF)}_{ij} \right)
    \end{align}
    with
    \begin{align}
        C^\text{geom}(\varepsilon, x) &= \frac{5 \vol(S^{d - 2})}{12 (d^2 - 1)} \int \dd t \dd r ~ (t^2 - r^2) r^d \abs{T^{(-1)}}^2 \Re[T^{(-1)} \complexconj{T^{(0)}}] \,.
    \end{align}
    Applying that $R^{(TF)}$ is trace-free to the second equation gives $R^{(TF)}_{\mu \nu} = 0$ for $\mu = \nu$.
    Combining these two equations gives $R^{(TF)} = 0$.

    For the matter perturbation we get
    \begin{align}
        0 &= \frac{T^{(TF)}_{\mu \nu}(x)}{8 \pi} \int \dd^d \xi \:\abs{T^{(-1)}}^2 \frac{\xi_\alpha \xi^\alpha}{r^2} \Re \left[ T^{(-1)} \right] \nonumber \\
        &\quad\quad \cdot \left( (\xi_\beta \xi^\beta) \chi^\mu \sigma^\nu - t \xi^\mu \xi^\nu \right) \left( \chi_\rho \xi_\sigma - \xi_\rho \chi_\sigma \right) \,.
    \end{align}
    Considering both perturbations, we get
    \begin{align}
        0 &= C^\text{matter}(\varepsilon, x) \, T^{(TF)}_{\mu \nu}(x) - C^\text{geom}(\varepsilon, x) \, R^{(TF)}_{\mu \nu}(x)
    \end{align}
    with
    \begin{align}
        C^\text{matter}(\varepsilon, x) &= \frac{\vol(S^{d - 2})}{4 \pi (d^2 - 1)} \int \dd t \dd r ~ (t^2 - r^2) r^d \abs{T^{(-1)}}^2 \Re[T^{(-1)}] \,.
    \end{align}

    Therefore, we recover the trace-free contribution to the Einstein equation
    \begin{align}
        R^{(TF)}_{\mu \nu}(x) = 8 \pi G(x) T^{(TF)}_{\mu \nu}(x) \,.
    \end{align}
    Note that the coupling parameter $G(x)$ becomes constant, i.e.~independent of $x$ if the regularizing vector field $\chi$
    only depends on the vector $\xi$ and not the base point $x$.
    Moreover, due to the appearance of the $T^{(0)}$ term in the geometric constant $C^\text{geom}$, $G$ is of order $\mathcal{O}(\varepsilon^2)$.

    \subsection*{Competing interests}
    The authors declare that they have no conflicts of interest to disclose.

    \subsection*{Data availability statement}
    The authors declare that the manuscript has no associated data.

    \bibliographystyle{amsplain}
%    \bibliography{../cfs-bibliography/main}
    \bibliography{main}
\end{document}